\documentclass[journal]{IEEEtran}
\usepackage{multirow}
\usepackage{setspace}
\usepackage{amsmath}
\usepackage{graphicx}
\usepackage{amssymb}
\usepackage{color}
\usepackage{longtable}
\usepackage{float}
\usepackage{epstopdf}
\usepackage{multirow}
\usepackage{bm}
\usepackage{cite}

\graphicspath{{figure/}}

\newtheorem{defi}{\bf Definition}
\newtheorem{rmk}{\bf Remark}

\newtheorem{lma}{\bf Lemma}
\newtheorem{thm}{\bf Theorem}

\begin{document}

\title{Method of Reduction of Variables for Bilinear Matrix Inequality Problems in System and Control Designs}

\author{Wei-Yu~Chiu,~\IEEEmembership{Member,~IEEE}
\thanks{This work was supported by the Ministry of Science and Technology of Taiwan
under Grant 102-2218-E-155-004-MY3.}%
\thanks{W.-Y. Chiu is with the Multiobjective Control Laboratory, Department of Electrical Engineering, Yuan Ze University, Taoyuan 32003, Taiwan
        (email: chiuweiyu@gmail.com).}%
\thanks{
\copyright 2016 IEEE. Personal use of this material is permitted. Permission from IEEE must be
obtained for all other uses, in any current or future media, including
reprinting/republishing this material for advertising or promotional purposes, creating new
collective works, for resale or redistribution to servers or lists, or reuse of any copyrighted
component of this work in other works.
}
\thanks{Digital Object Identifier 10.1109/TSMC.2016.2571323}               
 }
\maketitle

\begin{abstract}
Bilinear matrix inequality (BMI) problems in system and control designs are investigated
in this paper. A solution method of reduction of variables (MRV) is proposed.
This method consists of a principle of variable classification,
a procedure for problem transformation, and a hybrid algorithm that combines deterministic and stochastic search engines.
The classification principle is used to classify the decision variables of a BMI problem into two categories:  external and internal variables.
Theoretical analysis is performed to show that when the classification principle is applicable,
a BMI problem can be transformed into an unconstrained optimization problem that has fewer decision variables.
Stochastic search and deterministic search are then applied to determine the decision variables of the unconstrained problem externally and explore the internal problem structure, respectively.
The proposed method can address feasibility, single-objective, and multiobjective problems constrained by BMIs
in a unified manner.  A  number of numerical examples in system and control designs
 are provided to validate the proposed methodology.
Simulations show that the MRV can outperform existing BMI solution methods in most benchmark problems and achieve similar levels of performance in the remaining problems.
\end{abstract}


\section{Introduction}\label{sec_intro}

Bilinear matrix inequality (BMI) problems  frequently arise
in system and controller designs~\cite{vanantwerp2000tutorial}, e.g.,  low-authority controller (LAC) designs~\cite{SSS,LAC1},
static output feedback designs for spectral abscissa optimization/$H_2$ optimization/$H_\infty$ optimization~\cite{dinh2012combining,dinh2012inner,COMPleib},
affine fuzzy system designs~\cite{SAFS1,SAFS2}, and observer-based robust controller designs~\cite{OCS}.
The advantages of using BMI formulations can be observed in various scenarios.
For instance, BMI formulations can
avoid a nonsmooth objective function that is hard to handle  when spectral abscissa optimization is considered~\cite{burke2002two}; they may outperform linear matrix inequality (LMI) approaches that can fail to predict the stability of Takagi--Sugeno fuzzy systems~\cite{AM}; and they can
yield less conservative designs than using LMI formulations~\cite{hu2006stability}.

While BMI problems are NP-hard~\cite{toker1995np,NP1}, BMI solution methods are continuously investigated
in the literature because of the advantages derived from using BMI formulations.
In~\cite{SSS,ostertag2008improved,path2}, path-following methods were proposed in which
controller gains were  iteratively perturbed to achieve desired performance specifications.
The methods were based on the assumption that closed- and open-loop systems were  slightly different, i.e.,
LAC designs were considered.
In~\cite{dinh2012combining}, convex--concave decomposition and linearization methods (CCDM) were combined to address static output feedback problems.
After decomposition and linearization, BMI constraints were addressed by solving a sequence of convex semi-definite programming problems.
In~\cite{dinh2012inner}, an inner convex approximation method (ICAM) was proposed as a generalized version of the CCDM.
Nonlinear semi-definite programming was considered and a regularization technique was employed to ensure a strict descent search direction.
In~\cite{orsi2006newton}, a Newton-like search method closely related to alternative projection methods
was proposed to improve convergence properties.

Alternating minimization (AM) is another popular solution method and has been widely used because of their simplicity and effectiveness~\cite{AM,goh1994biaffine,el1994synthesis}.
For the AM methods, decision variables are divided into two groups. By fixing one group of variables, the other group of variables forms an LMI problem (LMIP), which is convex and can be solved efficiently. Decision variables in separate groups are then determined alternately during the solving process of LMIPs.
Variant versions include iterative LMI (ILMI) methods~\cite{SAFS1,cao1998static,zheng2002heuristic} and the two-step procedure~\cite{OCS}.
A few Matlab toolboxes for BMI problems are also available online. For example,
 LMIRank can be used to solve rank constrained LMI problems~\cite{LMIRank}.
HIFOO employs quasi-Newton updating and gradient sampling to search for solutions. It mainly focuses on fixed-order stabilization and performance optimization problems~\cite{HIFOO_soft,HIFOO,HIFOO2}.
PENBMI, commercial software, aims at solving BMI constrained optimization problems or optimization problems that have quadratic cost functions~\cite{henrion2005solving,PENBMI_soft}.

The aforementioned methods and software packages serve as local optimization approaches to BMI problems.
Because a BMI problem is nonconvex, local optima exist and, hence,
 local optimization approaches may not be able to achieve global optimality.
To avoid attaining local optimality, we consider global optimization approaches that employ heuristic algorithms.
In~\cite{goh1994global,LPVS,apkarian2000robust,tuan1999new,safonov1994control,fukuda2001branch},
 branch-and-bound (BB) type methods were proposed.
The BB type methods replace bilinear terms with bounded new variables so that a BMI problem can be relaxed into an LMIP.
Although being possible to achieve the global optimum,  BB type methods can bear a computational burden
 because the size of the LMIs that must be solved for the lower bound can increase exponentially
 upon increasing the number of decision variables~\cite{zheng2002heuristic}.
In~\cite{beran1997global}, another global optimization approach using generalized Benders decompositions was proposed for BMI problems,
but its performance was not evaluated through a number of test problems.

In general, existing BMI solution methods can suffer from at least one of the following five
drawbacks or limitations.
First, decision variables are expressed solely in a vector form, e.g., some BB type methods. By contrast, a matrix form is more convenient in control problems~\cite{Boyd_LMI}.
Second, solution methods are originally designed to fit particular problem structures.
In some situations, applying developed methods to other problem structures, if not impossible, requires extra efforts to reformulate the problem, e.g., some AMs and ILMI methods.
In other situations, solution methods cannot be applied to problems that do not have the intended structures, e.g., path-following methods.
Third, prior derivations such as approximations or decompositions must be performed before algorithms are applied, e.g., the CCDM and ICAM, and these derivations can be cumbersome and sometimes heuristic. Fourth, only local optimization is performed while BMI problems inherently have multiple local optima.
Finally, to the best of our knowledge, existing BMI solution methods cannot address multiobjective optimization problems (MOPs)
in which a set of Pareto optimal solutions is of interest rather than the global optimal solution.\footnote{The ability to solve MOPs constrained by BMIs is worth further investigation because
MOPs naturally and frequently arise in engineering problems~\cite{Category_1,7287774,6918520,6626653}.
Solving an MOP, yielding an approximate Pareto front (APF)  and Pareto optimal set, can provide a system designer with a broad perspective on optimality. The resulting APF
can clearly illustrate how one objective affects the others, and the obtained Pareto set
allows the designer to make a posterior decision, i.e., selecting  design parameters after a set of promising candidates is available~\cite{Boyd_linear,TEVC_17}. In general, a posterior decision is preferred to a prior decision because more information has been used before the decision making~\cite{B_MOEA1}.}

To avoid the aforementioned five drawbacks or limitations,
we propose a method of reduction of variables~(MRV). The method consists of a principle of variable classification,
a transformation of the BMI problem, and a hybrid multiobjective immune algorithm (HMOIA) that solves the problem derived from the transformation.
Internal and external  variables are coined and used to denote all the decision variables involved.
The internal variable can represent a set of  matrix variables,  which is convenient in controller designs.
To develop a general-purpose solution method, we consider possible multiple objectives in BMI problems  and assume no particular problem structures.
This yields a framework that addresses feasibility problems, single-objective optimization problems (SOPs), and MOPs constrained by BMIs in a unified manner.
The developed HMOIA is a hybrid because it employs  stochastic and deterministic mechanisms to determine the external and internal decision variables, respectively.
The stochastic mechanism allows for global exploration of the entire solution space.
By applying the HMOIA to BMI problems, few prior derivations, involving only variable classification and simple problem transformation, are required.
Limited derivations render the proposed method suitable for various BMI problems.

To verify the effectiveness of the MRV, we used a series of test problems in our simulations~\cite{COMPleib,AM,SAFS1,SAFS2,OCS,LPVS,SSS,ostertag2008improved}.
For feasibility problems, while different solution methods were developed to address various BMI problems,
  the MRV was able to find a solution with $100\%$ success rates in a unified manner.
In spectral abscissa optimization, the MRV outperformed existing methods in $73\%$ of selected benchmark problems in terms of the minimum value or mean value.
 The MRV achieved better levels of performance than existing methods in $27.5\%$ and $47.8\%$ of selected $H_2$ and $H_\infty$ optimization problems, respectively,
while it yielded similar performance in the remaining problems.
As shown in~\cite{dinh2012combining} and~\cite{dinh2012inner},
 the CCDM and ICAM were relatively robust compared with other existing solution methods.
We illustrated that the MRV was able to find solutions to certain problems in which
these two robust methods failed or made little progress towards a local solution.

The main contributions of this paper are as follows.
We propose a novel global optimization approach to BMI problems, which has  not been fully investigated compared to local optimization approaches.
This approach can combat a few drawbacks  existing BMI solution methods can suffer from: using inconvenient variable expression,  being confined to particular problem structures, requiring heuristic or cumbersome prior derivations, or being incapable of addressing multiple objectives.
When the proposed classification principle is applicable, we provide a unified formulation that facilitates generating solutions to feasibility problems,  SOPs, and MOPs constrained by BMIs. To the best of our knowledge, this is the first study that provides such a unified framework.
We perform related analysis and validate the proposed MRV through a large number of benchmark problems, showing that
the proposed methodology can outperform existing solution methods in many of these BMI problems.

The rest of this paper is organized as follows.
Section~\ref{sec_notation_prob} describes the problem formulation and the principle of variable classification.
In Section~\ref{sec_thm}, preliminaries to our algorithm development are examined, including analysis of problem transformations.
Section~\ref{sec_alg} presents the HMOIA  and hence, the MRV.
Simulation results are given in Section~\ref{sec_sim}. Finally,
Section~\ref{sec_con} concludes this paper.

\section{Problem Formulation and Variable Classification}\label{sec_notation_prob}

In this section, we investigate system and control designs that are formulated as BMI problems, and propose a classification principle for decision variables that facilitates
solution search.
Under our framework, the associated cost function can be a vector-valued function, a scalar function, or a constant,
 depending on the number of  objectives involved.
By using the classification principle,
 decision variables in  BMIs are classified into two types, the internal and external variables.
Design examples are presented to illustrate how to use the proposed principle of variable classification.

The following notation and terminology are used throughout this study.
Let $\mathbb{R}$ and $\mathbb{C}$ be the sets of real and complex numbers, respectively.
For a  scalar $b\in \mathbb{C} $, $\overline{b}$ denotes the complex conjugate of $b$.
Let $[\bm{a}]_i$ and $[\bm{A}]_{ij}$ denote the~$i$th entry of the vector~$\bm{a}$ and the~$(i,j)$th entry of the matrix~$\bm{A}$, respectively.
For two vectors $\bm{a}$ and $\bm{b}$,
  $\bm{a}\leq \bm{b}$ is interpreted as  $[\bm{a}]_i\leq [\bm{b}]_i$ for all $i$.
 If $\bm{P}>0$, then
$\bm{P}$ is symmetric and positive-definite.
Similarly, $\bm{P}<0$ implies that
$\bm{P}$ is symmetric and negative-definite.
For a square matrix~$\bm{A}$,  $eig(\bm{A})$ represents the vector of all eigenvalues of $\bm{A}$ placed in a prescribed manner, and  $eig\{\bm{A}\}$ represents the set of all eigenvalues of $\bm{A}$.
The mark ``$\star$'' is used to denote the induced symmetry, e.g., $(\bm{PA},\star)=\bm{PA}+\bm{A}^T \bm{P}^T$ and
\begin{equation*}
    \left[
      \begin{array}{cc}
        \bm{A} & \bm{B}^T \\
        \bm{B} & \bm{C} \\
      \end{array}
    \right]
    =
    \left[
      \begin{array}{cc}
        \bm{A} & \star \\
        \bm{B} & \bm{C} \\
      \end{array}
    \right]
   =
       \left[
      \begin{array}{cc}
        \bm{A} & \bm{B}^T \\
        \star & \bm{C} \\
      \end{array}
    \right].
\end{equation*}
If $\bm{f}:\Omega  \rightarrow \mathbb{R}^N$ is a vector-valued function, then the MOP
\begin{equation}\label{eq_def_MOP}
 \begin{split}
   \min_{\bm{\omega}} \; &  \bm{f}(\bm{\omega})  \\
  \mbox{subject to } &  \bm{\omega} \in \Omega
 \end{split}
\end{equation}
is interpreted as vector optimization in which Pareto optimality is adopted.
The domain $\Omega$ lies in the Euclidean space $\mathbb{R}^M$ for some positive integer $M$.
The associated terminology  is presented as follows~\cite{IET_CTA_14,chiu_MO_filter,TSG_15}.

\begin{defi}[Pareto dominance]\label{def_domi}
In the decision variable space of~~(\ref{eq_def_MOP}), a point
$\bm{\omega}' \in \Omega$ dominates another point $\bm{\omega}'' \in \Omega$
if  the conditions  $[\bm{f}(\bm{\omega}')]_i  \leq [\bm{f}(\bm{\omega}'')]_i,i=1,2,...,N,$
hold true and at least one inequality is strict.
 In this case, we denote  $\bm{\omega}' \preceq_{\bm{f}}   \bm{\omega}'' $ and $\bm{f}(\bm{\omega}')\preceq \bm{f}(\bm{\omega}'')$.
A point that is not dominated by other points is termed a nondominated point.
\end{defi}

\begin{defi}[Pareto optimal set]\label{def_P_opt}
 The Pareto optimal set $\mathcal{P}^*$ of~(\ref{eq_def_MOP}) is defined as the set of all nondominated points, i.e.,
 \begin{equation*}
  \mathcal{P}^*=  \{  \bm{\omega} \in \Omega : \nexists \bm{\omega}' \in  \Omega  \mbox{ such that } \bm{\omega}' \preceq_{\bm{f}}
   \bm{\omega} \}.
 \end{equation*}
\end{defi}

\begin{defi}[Pareto front]\label{def_PF}
The Pareto front (PF) of~(\ref{eq_def_MOP}) is defined as the image of the Pareto optimal set through the mapping $\bm{f}$, i.e., $\bm{f}(\mathcal{P}^*)$ represents the PF.
\end{defi}

In our BMI-based  design problems, we use
\begin{equation}\label{eq_def_BMI_general}
  \mathcal{BMI}( \bm{\alpha},\bm{X} ) < 0
\end{equation}
to represent a BMI, where $\mathcal{BMI}( \cdot )$ is a matrix function, and
 $\bm{\alpha}$ and $\bm{X}$ are the variables.
The inequality $\mathcal{BMI}( \bm{\alpha},\bm{X} ) < 0$
becomes an LMI in the variable  $\bm{\alpha}$ given $\bm{X}$  or in the variable $\bm{X}$ given $\bm{\alpha}$.
 If more than one BMI are involved, then the notation $\mathcal{BMI}( \bm{\alpha},\bm{X} )$ represents a block-diagonal matrix such that $ \mathcal{BMI}( \bm{\alpha},\bm{X} ) < 0$ consists of all the BMIs.

To consider optimal designs in a unified framework, we  add an objective function $\bm{\mathcal{F}}(\cdot)$ to~(\ref{eq_def_BMI_general}).
From the perspective of algebra, there is no difference between $ \bm{\alpha}$ and $\bm{X} $ in~(\ref{eq_def_BMI_general})
  because they are just two coupled  variables in the BMI $\mathcal{BMI}( \bm{\alpha},\bm{X} ) < 0$.
However, to create a solution method, we
assume $\bm{\mathcal{F}}$ is a function of $\bm{\alpha}$.
The resulting BMI-based MOP can be expressed as
\begin{equation}\label{eq_def_BMI}
 \begin{split}
   \min_{\bm{\alpha},\bm{X}} \; &  \bm{\mathcal{F}}(\bm{\alpha} )  \\
  \mbox{subject to }&  \mathcal{BMI}( \bm{\alpha},\bm{X} ) < 0
 \end{split}
\end{equation}
where $ \bm{\alpha}$ is distinguished  from $\bm{X} $ by using the following  classification principle.

\textbf{Principle of Variable Classification}:
\begin{enumerate}
  \item Upper and lower bounds on the entries of variables in $\bm{\alpha}$ are available or can be obtained.
      Square matrix variables in~$\bm{\alpha}$, if any, do not have constraints on definiteness, i.e., positive or negative definiteness.
  \item Bounds on entries of variables in $\bm{X}$ are unavailable.
  \item The objective function $\bm{\mathcal{F}}$ can be expressed solely in terms of $\bm{\alpha}$.
  \item The size of~$\bm{\alpha}$ should be as small as possible.
\end{enumerate}

When the classification principle is applicable,
we term the variables $ \bm{\alpha}$ and $\bm{X}$  the external and internal decision variables, respectively.
In our principle, $\bm{\alpha}$ represents those variables (scalar and/or matrix variables) in a BMI problem that have bounds on entries. These bounds are mostly inherent from physical constraints or can be readily assigned mathematically.
 The remaining variables (scalar and/or matrix variables) are included in $\bm{X}$.
  They generally do not have upper and lower bounds on their entries,
   but there can be constraints related to positive or negative definiteness imposed on matrix variables in $\bm{X}$.
 The definiteness associated with matrix variables in $\bm{X}$ is required to ensure the system stability, which mainly distinguishes $\bm{X}$ from $\bm{\alpha}$.
A typical external variable can include controller gains and/or system parameters.
By contrast, matrix variables related to the Lyapunov theory are classified as the internal variable because bounds on the entries of these matrix variables are unavailable in practice.

The condition in which $\bm{\mathcal{F}}$ is not a function of $\bm{X}$ does not
yield a restricted problem formulation. For instance, if  $[\bm{\mathcal{F}}]_i= g(\bm{X})$  is encountered,
we may introduce a slack variable $\eta$, impose the constraint $g(\bm{X})\leq \eta$, and assign $[\bm{\mathcal{F}}]_i:= \eta$.
In this way, the objective function becomes the one with  $\bm{\alpha}$ as the only variable.
Finally, it will be shown that the BMI-constrained problem in~(\ref{eq_def_BMI})
can be reduced to an unconstrained problem in which  $\bm{\alpha}$  is the only decision variable.
Therefore, a smaller size of~$\bm{\alpha}$ means the fewer number of decision variables in the
unconstrained problem, which explains why we keep the size of~$\bm{\alpha}$ as small as possible in the classification principle.

When $\bm{\mathcal{F}}(\cdot)$ is a constant function,
it is understood that  the BMI problem in~(\ref{eq_def_BMI}) is interpreted as a feasibility problem.
Otherwise, an SOP (or MOP) is considered if
 $\bm{\mathcal{F}}(\cdot)$ is a scaler-valued (or vector-valued) function.
For a feasibility problem, it is desired to determine whether or not there exists a point~$( \bm{\alpha},\bm{X} ) $ satisfying the matrix inequality~$\mathcal{BMI}( \bm{\alpha},\bm{X} ) < 0$. If such a point exists, then the problem is feasible and any point that satisfies the matrix inequality
is  a solution (or a feasible point).
For an SOP, it is desired to search for a feasible point that achieves the minimum value of the objective function.
When an MOP is considered, the associated optimality is interpreted as Pareto optimality. In that case, the Pareto optimal set is to be determined.

To illustrate how to use the principle of variable classification, we examine a few
design examples as follows.

\subsection{Feasibility Problems}

\emph{Stability Test (ST)}: Consider a T--S fuzzy system~\cite{model_ST}
\begin{equation}\label{eq_sys_ST}
    \dot{\bm{x}}(t)= \sum_{i=1}^2 \xi_i(\bm{x}(t)) \bm{A}_i \bm{x}(t) +\bm{p}, \bm{p}^T \bm{p} \leq  \mu^2 \bm{x}(t)^T \bm{x}(t).
\end{equation}
It can be shown that the system in~(\ref{eq_sys_ST}) is stable if there exist $\tau_{\ell i j}\geq 0$ and $\bm{P}_i>0$ such that~\cite{AM}
\begin{equation}\label{eq_ST_BMI}
   \bm{A}_\ell^T \bm{P}_i+\bm{P}_i \bm{A}_\ell+ \mu^2 \bm{I}-  \sum_{j=1}^2 \tau_{\ell i j} (\bm{P}_j-\bm{P}_i) <0, \mbox{ for } \ell,i=1,2
\end{equation}
are satisfied.
According to the classification principle, the external variable cannot include matrix variables that have a constraint on definiteness.  Because $\bm{P}_i>0,i=1,2,$ are positive-definite matrix variables, they must be included in the internal variable $\bm{X}$; to yield a BMI problem, the remaining variables  $\tau_{\ell i j}$ are included in the external variable $\bm{\alpha}$.
The feasibility problem in~(\ref{eq_ST_BMI}) can then be expressed as~$\mathcal{BMI}( \bm{\alpha},\bm{X} ) < 0$ in which
$\bm{\alpha}=(\tau_{112},\tau_{121},\tau_{212},\tau_{221})$ and $\bm{X}=(\bm{P}_1,\bm{P}_2)$.

\subsection{Single-objective Optimization Problems}\label{sub_SOP}

\emph{Linear Parameter-varying Systems (LPVS)}:  Consider a linear time-varying system~\cite{model_LPVS1,model_LPVS2}
\begin{equation}\label{eq_sys_LPVS}
    \dot{\bm{x}}(t) =  \bm{A}(t) \bm{x}(t) ,\bm{A}(t) \in \mbox{ convex hull}\{\bm{A}_1,\bm{A}_2 \}
\end{equation}
where
\begin{equation*}
    \bm{A}_1=
    \left[
      \begin{array}{cc}
        0 & 1 \\
        -2 & -1 \\
      \end{array}
    \right] \mbox{ and }
    \bm{A}_2=
    \left[
      \begin{array}{cc}
        0 & 1 \\
        -2-\varsigma & -1 \\
      \end{array}
    \right].
\end{equation*}
The $\varsigma$ represents a design parameter.
The system in~(\ref{eq_sys_LPVS}) is stable if there exist $\delta_i$ and $\bm{P}_i$ satisfying~\cite{LPVS,LPVS2}
\begin{equation}\label{eq_LPVS_BMI}
 \begin{split}
 &     (1-\delta_2)(\bm{P}_2\bm{A}_1,\star)+ \delta_2 (\bm{P}_2-\bm{P}_1)<0 \\
 & (1-\delta_1)(\bm{P}_1\bm{A}_2,\star)- \delta_1 (\bm{P}_2-\bm{P}_1)<0 \\
  & (\bm{P}_1\bm{A}_1,\star)<0,(\bm{P}_2\bm{A}_2,\star)<0 \\
  & 0<\bm{P}_i<\bm{I},  0\leq \delta_i \leq 1, i=1,2.
 \end{split}
\end{equation}
For a fixed~$\varsigma$,~(\ref{eq_LPVS_BMI}) is a BMI in the variables $(\delta_1,\delta_2)$ and $(\bm{P}_1,\bm{P}_2)$.
To find the largest value of $\varsigma$ yielding a stable system, we can
 solve
\begin{equation}\label{eq_LPVS_SO}
 \begin{split}
  \max_{\varsigma,\delta_i,\bm{P}_i  } \;  &  \varsigma
     \\
  \mbox{subject to } &  (\ref{eq_LPVS_BMI}).
 \end{split}
\end{equation}
Based on the principle of variable classification,
$\bm{P}_1$ and $\bm{P}_2$ are positive-definite and must be included in the internal variable~$\bm{X}$;
to have $\mathcal{BMI}( \bm{\alpha},\bm{X} ) < 0$ as an LMI problem for a fixed $\bm{\alpha}$,
we are forced to include all the remaining variables in the external variable.
We thus have  $\bm{\alpha}=(\varsigma,\delta_1,\delta_2)$, $\bm{X}=(\bm{P}_1,\bm{P}_2)$, and $\bm{\mathcal{F}}(\bm{\alpha} )= -\varsigma$. The negative sign in $\bm{\mathcal{F}}$ has been added for the conversion of~(\ref{eq_LPVS_SO}) to the minimization form of~(\ref{eq_def_BMI}).

\subsection{Multiobjective Optimization Problems}

For a sparse linear constant output-feedback design,
the BMI problem~\cite{SSS,dinh2012combining}
\begin{equation}\label{eq_sparse}
 \begin{split}
    \min_{\beta,\bm{F},\bm{P}  } & \; -\sigma \beta +   \sum_i \sum_j  |[\bm{F}]_{ij}|  \\
    \mbox{subject to }  & \;   ( \bm{P} \bm{A}_{\bm{F}},\star)+ 2\beta \bm{P} <0 , \bm{P} >0\\
 \end{split}
\end{equation}
can be formulated, where $\bm{A}_{\bm{F}}=\bm{A}+\bm{BFC}$, $\sigma>0$ represents a prescribed weighting coefficient, and $\beta $ represents the decay rate.
The SOP in~(\ref{eq_sparse}) is interpreted as determining the controller gain~$\bm{F}$ so that
the decay rate $\beta $ is maximized and~$\bm{F}$ is kept as much sparse as possible.
One drawback of considering the single-objective formulation is that  there is no rule that can be used to assign the value of $\sigma$, which affects the values of $\beta$ and~$\bm{F}$.
In practice,
 a system designer  selects an arbitrary value of~$\sigma$  and accepts the resulting gain~$\bm{F}$.
To avoid such heuristic assignment for~$\sigma$, we can consider a multiobjective formulation that addresses two objectives in separate dimensions~\cite{B_MOEA1,B_MOEA2}:
\begin{equation}\label{eq_sparse_mod}
 \begin{split}
    \min_{\beta,\bm{F},\bm{P}  } & \;
    \left[
      \begin{array}{cc}
      -\beta  &   \sum_i \sum_j  |[\bm{F}]_{ij}|  \\
      \end{array}
    \right]^T  \\
    \mbox{subject to }  & \;   ( \bm{P} \bm{A}_{\bm{F}},\star)+ 2\beta \bm{P} <0 , \bm{P} >0.
 \end{split}
\end{equation}
According to the classification principle, $\bm{P}$ is positive-definite and hence, included in the internal variable~$\bm{X}$;
to yield a BMI problem, the remaining variables must be included in the external variable~$\bm{\alpha}$.
Referring to~(\ref{eq_def_BMI}),  we have~$\bm{\alpha}=(\beta,\bm{F})$, $\bm{X}=\bm{P}$, and
$\bm{\mathcal{F}}(\bm{\alpha} )=[   -\beta  \;   \sum_i \sum_j  |[\bm{F}]_{ij}|]^T$.
Once~(\ref{eq_sparse_mod}) has been solved, an approximate Pareto front~(APF) can be obtained and the system designer can select an appropriate~$\bm{F}$ based on the information provided by the APF.

The proposed classification principle is based on the basic properties of BMIs represented by $ \mathcal{BMI}( \bm{\alpha},\bm{X} ) < 0$. Because BMIs are nonlinear and have possibly several local optima when optimization is involved,
any deterministic algorithms can be trapped locally. To remedy this problem, stochastic algorithms can be used.
 Because $ \mathcal{BMI}( \bm{\alpha},\bm{X} ) < 0$ is a BMI,
  $ \mathcal{BMI}( \bm{\alpha},\bm{X} ) < 0$ becomes an LMI in the  variable $\bm{X}$ for a fixed value of $\bm{\alpha}$.
  For LMIs, it is well-known that deterministic algorithms such as interior-point methods are suitable for solving them efficiently. These arguments suggest that variables in BMI problems be classified into two groups so that a hybrid algorithm combining stochastic and deterministic search engines can be applied.

To integrate stochastic and deterministic search schemes, we first
explore the variable space of $\bm{\alpha}$  (external exploration) so that an LMIP  $ \mathcal{BMI}( \bm{\alpha},\bm{X} ) < 0$ in the variable $\bm{X}$ can be obtained. Once $\bm{\alpha}$ is determined, the associated variable space of $\bm{X}$ can then be searched internally and efficiently because of the convexity.
  This explains why  $\bm{\alpha}$ and $\bm{X}$ are termed external and internal variables, respectively.
Since $\bm{X}$ is relevant to the feasibility but irrelevant to the objective values,
 this internal variable can be considered hidden from the external search if
 information about the feasibility is extracted properly.
Therefore, we may reduce the original problem with variables $\bm{\alpha}$ and $\bm{X}$ to a simpler problem with only the variable $\bm{\alpha}$, and then transform the resulting problem into another form that is convenient for addressing the feasibility condition.

\section{Preliminaries to Algorithm Development}\label{sec_thm}

This section discusses the reduction and transformation, and other preliminary results that are helpful in later development of the hybrid algorithm. The section is divided into three subsections:
Section~\ref{subsec_thm} focuses on
theorems that transform the BMI problem in~(\ref{eq_def_BMI}) into an unconstrained problem with fewer decision variables;
Section~\ref{subsec_L-M} presents a solution method related to pole placement problems; and Section~\ref{subsec_density} describes an algorithm that reduces the population density of the HMOIA.

\subsection{Reduction and Equivalence Theorems}\label{subsec_thm}

Theorems in this subsection lead to an optimization problem that has a simpler form than~(\ref{eq_def_BMI}).
By using the theorems, the number of decision variables in~(\ref{eq_def_BMI}) can be reduced, and the associated problem
 can be further  transformed into an unconstrained optimization problem.
Although we adopt multiobjective formulations in the following discussions,
 the established results remain true when an SOP or a feasibility problem is considered.

Consider the eigenvalue problem (EVP)
\begin{equation}\label{eq_EVP}
 (\lambda^*( \bm{\alpha}),  \bm{X}^*( \bm{\alpha} ))
  =
  \begin{array}{l}
    \arg_{\lambda,\bm{X}}  \min_{\lambda,\bm{X}} \;  \lambda \\
    \mbox{subject to }  \mathcal{BMI}( \bm{\alpha},\bm{X} )  <  \lambda \bm{I}.
  \end{array}
\end{equation}
Because the constraint $\mathcal{BMI}( \bm{\alpha},\bm{X} ) < 0$
is a BMI, the EVP in~(\ref{eq_EVP}) is convex in the variables
$\lambda$ and $\bm{X}$
given the value of~$\bm{\alpha}$. (For a fixed value of $\lambda$,  the EVP can thus be solved by interior-point methods.)
In~(\ref{eq_EVP}), we denote~$ (\lambda^*( \bm{\alpha}),  \bm{X}^*( \bm{\alpha} ))$ as the pair that achieves the minimum.
Both $\lambda^*( \bm{\alpha})$ and $\bm{X}^*( \bm{\alpha} )$ are regarded as a function of~$\bm{\alpha}$.
The following lemma relates the value of $\lambda^*( \bm{\alpha})$ to the feasibility of~(\ref{eq_def_BMI}).

\begin{lma}\label{lma_EVP}
The BMI problem in~(\ref{eq_def_BMI}) is feasible if and only if  an $\tilde{\bm{\alpha}}$ exists
such that the value of $\lambda^*( \tilde{\bm{\alpha}})$ in~(\ref{eq_EVP}) is negative, i.e., $\lambda^*( \tilde{\bm{\alpha}})  <0$.
\end{lma}
\begin{proof}
It can be readily verified by a slight modification of the proof in Lemma~1 of~\cite{IET_CTA_14}  or~\cite{CACS_14}.
\end{proof}

The following theorem follows from using Lemma~\ref{lma_EVP}.

\begin{thm}[Reduction Theorem]\label{thm_reduce}
There exists a pair $(\tilde{\bm{\alpha}},\tilde{\bm{X}})$ that is Pareto optimal in~(\ref{eq_def_BMI}) if and only if (denoted by $\Leftrightarrow$) $\tilde{\bm{\alpha}}$ is Pareto optimal in
\begin{equation}\label{eq_BMI_reduced}
 \begin{split}
   \min_{\bm{\alpha}} \; &   \bm{\mathcal{F}}(\bm{\alpha} ) \\
  \mbox{subject to }&   \lambda^*(\bm{\alpha})<0.
 \end{split}
\end{equation}
\end{thm}
\begin{proof}
We first prove necessity ($\Rightarrow$). By Lemma~\ref{lma_EVP}, we have $\lambda^*( \tilde{\bm{\alpha}})  <0$ and hence,
$\tilde{\bm{\alpha}}$ is a feasible point of~(\ref{eq_BMI_reduced}). Let us proceed by contraposition.
Suppose that there exists an $\bm{\alpha}'$ dominating $ \tilde{\bm{\alpha}}$ in~(\ref{eq_BMI_reduced}), i.e.,
\begin{equation}\label{eq_proof_red}
   \bm{\alpha}' \preceq_{\bm{\mathcal{F}}}  \tilde{\bm{\alpha}}, \lambda^*( \tilde{\bm{\alpha}})  <0, \mbox{ and }\lambda^*(\bm{\alpha}')  <0.
\end{equation}
However, the conditions in~(\ref{eq_proof_red}) implies that
\begin{equation*}
  ( \bm{\alpha}', \bm{X}^*(\bm{\alpha}') )\preceq_{\bm{\mathcal{F}}} (\tilde{\bm{\alpha}},\tilde{\bm{X}})
\end{equation*}
which yields a contradiction.

To prove sufficiency ($\Leftarrow$), we again use  contraposition. Suppose that there exists  a pair $(\bm{\alpha}',\bm{X}')$ dominating
$(\tilde{\bm{\alpha}},\tilde{\bm{X}})$, i.e.,
\begin{equation}\label{eq_proof_red2}
\small{
( \bm{\alpha}', \bm{X}' )\preceq_{\bm{\mathcal{F}}} (\tilde{\bm{\alpha}},\tilde{\bm{X}}), \mathcal{BMI}( \bm{\alpha}', \bm{X}' )  <0, \mbox{ and }\mathcal{BMI}(\tilde{\bm{\alpha}},\tilde{\bm{X}})<0 .}
\end{equation}
By Lemma~\ref{lma_EVP}, the conditions in~(\ref{eq_proof_red2}) are equivalent to those in~(\ref{eq_proof_red}), which implies that $\tilde{\bm{\alpha}}$ is not Pareto optimal in~(\ref{eq_BMI_reduced}). However, this contradicts the Pareto optimality of~$\tilde{\bm{\alpha}}$.
\end{proof}

According to Theorem~\ref{thm_reduce}, the BMI problem in~(\ref{eq_def_BMI}) with $\bm{\alpha}$ and $\bm{X}$
as the decision variables
 can reduce to~(\ref{eq_BMI_reduced}) with $\bm{\alpha}$  as the only decision variable.
That is, the number of decision variables is reduced, which explains why Theorem~\ref{thm_reduce} is termed the Reduction Theorem.
In the theorem, the expression ``Pareto optimal'' is replaced by
 ``feasible'' if we consider a feasibility problem. In this case,
the theorem is exactly the same as Lemma~\ref{lma_EVP}.
Similarly, we replace ``Pareto optimality'' with conventional optimality when an SOP is encountered.

By Theorem~\ref{thm_reduce}, we can solve~(\ref{eq_BMI_reduced}) for a BMI-based  design in place of~(\ref{eq_def_BMI}).
We  further consider an unconstrained problem
that is equivalent to~(\ref{eq_BMI_reduced}).

\begin{thm}[Equivalence Theorem]\label{thm_equiv}
Let
\begin{equation}\label{eq_F_tilde}
\tilde{ \bm{\mathcal{F}}}(\bm{\alpha} )=
     \left[
   \begin{array}{cc}
    \bm{\mathcal{F}}(\bm{\alpha} )^T   & \max\{0, \lambda^*(\bm{\alpha}) \}\\
   \end{array}
 \right]^T
\end{equation}
 where $\max\{0, \lambda^*(\bm{\alpha}) \}$ represents the maximum element in the set~$\{0, \lambda^*(\bm{\alpha}) \}$.
A point $\tilde{\bm{\alpha}}$ is Pareto optimal in~(\ref{eq_BMI_reduced})
 if and only if (denoted by $\Leftrightarrow$) $\tilde{\bm{\alpha}}$ satisfies the condition $\max\{0, \lambda^*(\tilde{\bm{\alpha}}) \}=0$ and is Pareto optimal in
 \begin{equation}\label{eq_BMI_equi}
   \min_{\bm{\alpha}} \; \tilde{ \bm{\mathcal{F}}}(\bm{\alpha} ).
\end{equation}
\end{thm}
\begin{proof}
We prove necessity ($\Rightarrow$).
Since $\tilde{\bm{\alpha}}$ is Pareto optimal in~(\ref{eq_BMI_reduced}), we have $\lambda^*(\bm{\alpha})<0$ and thus
$\tilde{\bm{\alpha}}$ satisfies the condition $\max\{0, \lambda^*(\tilde{\bm{\alpha}}) \}=0$.
We use contraposition.
Suppose that $\tilde{\bm{\alpha}}$ is not Pareto optimal in~(\ref{eq_BMI_equi}).
There must exist an $\bm{\alpha}'$ such that $\bm{\alpha}' \preceq_{\tilde{ \bm{\mathcal{F}}}}   \tilde{\bm{\alpha}}$, yielding
$\max\{0, \lambda^*(\bm{\alpha}') \}=0$.
However, this implies that the conditions in~(\ref{eq_proof_red}) hold true, i.e., $\tilde{\bm{\alpha}}$ is not Pareto optimal in~(\ref{eq_BMI_reduced}), which yields a contradiction.

We now prove sufficiency ($\Leftarrow$) and again use contraposition.
Suppose that
 an $\bm{\alpha}'$ exists such that   the conditions in~(\ref{eq_proof_red})
hold true. This implies $\bm{\alpha}'  \preceq_{\tilde{ \bm{\mathcal{F}}}}  \tilde{\bm{\alpha}}$, which yields a contradiction.
\end{proof}

Theorem~\ref{thm_equiv} is termed the Equivalence Theorem because it establishes an equivalence relation between~(\ref{eq_BMI_equi}) and~(\ref{eq_BMI_reduced}).
According to Reduction and Equivalence Theorems,
 we can solve the unconstrained problem in~(\ref{eq_BMI_equi}) that has
fewer decision variables
than the original BMI problem in~(\ref{eq_def_BMI}).

\subsection{Levenberg--Marquardt Method}\label{subsec_L-M}

Pole placement problems occur frequently in controller designs~\cite{pole_place,pole_place2}.
In this subsection, we investigate a trust region Levenberg--Marquardt method that can be used for pole placement
 when the system matrix $\bm{A}+\bm{BFC}$  is encountered. In this situation,  matrices $\bm{A}$, $\bm{B}$, and $\bm{C}$ are known, and $\bm{F}$ is the design parameter that must be determined.

Suppose that $\bm{A} \in \mathbb{R}^{n_x \times n_x  }$ and  $\bm{F} \in \mathbb{R}^{n_u \times n_y  }$  in which  $n_x,n_y,$ and $n_u$ represent the dimensions of the state vector, physical output, and control input, respectively.
To facilitate the following discussions, we reshape the gain matrix $\bm{F}$ into an $n_u n_y \times 1$ vector~$\bm{q}$,
and denote  $\bm{A}(\bm{q})=\bm{A}+\bm{BFC}\in \mathbb{R}^{n_x \times n_x  }$
and
\begin{equation}\label{eq_cost_L_M}
h(\bm{q},\bm{\lambda}^{pre})=\frac{1}{2}||eig(\bm{A}(\bm{q}))-\bm{\lambda}^{pre}  ||_2^2
\end{equation}
where $eig(\bm{A}(\bm{q}))$ represents the vector of eigenvalues of~$\bm{A}(\bm{q})$
and $\bm{\lambda}^{pre}$ is a prescribed vector of poles.
The entries of the vector~$eig(\bm{A}(\bm{q}))$ in~(\ref{eq_cost_L_M}) is placed in a way that the
minimum norm is achieved. The associated pole placement problem can be formulated as
\begin{equation}\label{eq_L_M_prob}
  \bm{q}^*(\bm{\lambda}^{pre}) =\arg_{\bm{q}} \min_{\bm{q}}\; h(\bm{q},\bm{\lambda}^{pre})
\end{equation}
which is an unconstrained nonlinear least squares problem.
To apply the trust region Levenberg--Marquardt method to solve~(\ref{eq_L_M_prob}),
we need the first partial derivatives  and an approximate Hessian matrix of $h(\bm{q},\bm{\lambda}^{pre})$ in~(\ref{eq_cost_L_M}).
For the vector of eigenvalues  $eig(\bm{A}(\bm{q}))$,
let $\bm{X}_e(\bm{q}) \in \mathbb{C}^{ n_x \times n_x }$ be the matrix consisting of the associated eigenvectors such that
\begin{equation*}
\begin{split}
   & \bm{A}(\bm{q})\bm{X}_e(\bm{q})\\
 { } = { }   & \bm{X}_e(\bm{q}) diag([eig(\bm{A}(\bm{q}))]_1,[eig(\bm{A}(\bm{q}))]_2,...,[eig(\bm{A}(\bm{q}))]_{n_x}).
\end{split}
\end{equation*}
We have~\cite{pole_derive,pole_derive2}
\begin{equation*}
    \frac{\partial [eig(\bm{A}(\bm{q}))]_i}{ \partial [\bm{q}]_m}=
    \left[
      \begin{array}{c}
     \bm{X}_e(\bm{q})^{-1}  \frac{\partial \bm{A}(\bm{q})  }{\partial [\bm{q}]_m }     \bm{X}_e(\bm{q})   \\
      \end{array}
    \right]_{ii}
\end{equation*}
for  $i=1,2,...,n_x$, and $m=1,2,...,n_u n_y$.
The first partial derivatives $\partial h(\bm{q},\bm{\lambda}^{pre})/\partial [\bm{q}]_m$  and approximate Hessian matrix  $\bm{H}$
can be expressed as
\begin{equation}\label{eq_1st_2nd_partial}
\begin{split}
\frac{\partial h(\bm{q},\bm{\lambda}^{pre}) }{\partial [\bm{q}]_m }   { }={ } &
Re
\left\{
  \begin{array}{c}
  \sum\limits_{i=1}^{n_x}  \overline{(  [eig(\bm{A}(\bm{q})) ]_i   -[\bm{\lambda}^{pre}]_i )}
    \\
  \end{array}
\right.
 \\
 &
\left.
  \begin{array}{c}
  \times \frac{\partial [eig(\bm{A}(\bm{q}))]_i}{ \partial [\bm{q}]_m}
    \\
  \end{array}
\right\}    \mbox{ and }\\
  [\bm{H}]_{m\ell} { }={ } &
Re
\left\{
  \begin{array}{c}
  \sum_{i=1}^{n_x} \overline{( \frac{\partial [eig(\bm{A}(\bm{q}))]_i}{ \partial [\bm{q}]_m})}    (\frac{\partial [eig(\bm{A}(\bm{q}))]_i}{ \partial [\bm{q}]_{\ell}})
    \\
  \end{array}
\right\}
\end{split}
\end{equation}
respectively.
The Levenberg--Marquardt algorithm for the pole placement problem in~(\ref{eq_L_M_prob}) is described as follows.\footnote{The reader can refer to Theorems~1--4 in~\cite{TR_1} or Theorems 4.8, 4.9, and 6.4 in~\cite{TR_2}
for the convergence analysis of the trust region method.}

\noindent\rule{8.8cm}{0.4pt}\\
\textbf{Trust Region Levenberg--Marquardt Algorithm}~\cite{TR_1,TR_2}\\
\noindent\rule{8.8cm}{0.4pt}\\
Given $\hat{\Delta}>0$, $\Delta_0\in (0,\hat{\Delta})$, and $\eta\in[0,1/4)$ \\
\textbf{For} $k=1,2,...$
\begin{itemize}
\item[]  Evaluate $\bm{p}_k$ by solving
\begin{equation}\label{eq_app_prop}
  \begin{split}
\bm{p}_k=\arg_{\bm{p}}  \min_{\bm{p}} \; &  m_k(\bm{p})  \\
  \mbox{subject to }&  || \bm{p} ||_2 \leq \Delta_k
 \end{split}
\end{equation}
where $\Delta_k$ represents the current trust region radius, and
\begin{equation}\label{eq_m_k}
    m_k(\bm{p})=h(\bm{q}_k,\bm{\lambda}^{pre})+\nabla h(\bm{q}_k,\bm{\lambda}^{pre})^T\bm{p}+  \bm{p}^T \bm{H}_k  \bm{p}
\end{equation}
with entries of $\nabla h(\bm{q}_k,\bm{\lambda}^{pre})$ and $\bm{H}_k $ defined in~(\ref{eq_1st_2nd_partial}).
\item[]  Evaluate
\begin{equation}\label{eq_rho_k}
\phi_k=
\frac{h(\bm{q}_k,\bm{\lambda}^{pre}) - h(\bm{q}_k+\bm{p}_k,\bm{\lambda}^{pre}) }{m_k(0)- m_k(\bm{p}_k)}.
\end{equation}
\item[] \textbf{If} $\phi_k<1/4$
           \begin{equation*}
            \Delta_{k+1}:= \Delta_{k}/4
           \end{equation*}
\textbf{Else}
 \begin{itemize}
 \item[]  \textbf{If} $\phi_k>3/4$ and  $||\bm{p}_k||_2=\Delta_k$
           \begin{equation*}
            \Delta_{k+1}:=\min \{2\Delta_{k},\hat{\Delta} \}
           \end{equation*}
           \textbf{Else}
             \begin{equation*}
            \Delta_{k+1}:=   \Delta_{k}
           \end{equation*}
           \textbf{End If}
  \end{itemize}
  \textbf{End If}
  \item[]  \textbf{If} $\phi_k>\eta$
           \begin{equation*}
           \bm{q}_{k+1}:= \bm{q}_{k} + \bm{p}_{k}
           \end{equation*}
           \textbf{Else}
           \begin{equation*}
            \bm{q}_{k+1}:= \bm{q}_{k}
           \end{equation*}
           \textbf{End If}
\end{itemize}
\textbf{ End For} \\
\noindent\rule{8.8cm}{0.4pt}

\subsection{Density Reduction Algorithm}\label{subsec_density}

When an evolutionary algorithm searches for Pareto optimal solutions to an MOP,  less crowded points must be preserved so that population diversity can be ensured.
To this end, we  estimate the density of current population and remove points that lie in a dense region.
We
 denote $\mathcal{A}(t_c)$ as the current population with the cardinality $|\mathcal{A}(t_c)|$,
$N_{nom}$ as the nominal size of the population, and $\bm{\alpha}$ as an element of $\mathcal{A}(t_c)$.
Suppose that $\tilde{ \bm{\mathcal{F}}}(\bm{\alpha} )\in \mathbb{R}^{N+1}$.  The process of removing points from a dense region is termed density reduction, which can be performed by the following algorithm modified from~\cite{AIS1}.

\noindent\rule{8.8cm}{0.4pt}\\
\textbf{Density Reduction Algorithm}\\
\noindent\rule{8.8cm}{0.4pt}\\
\textbf{While} $|\mathcal{A}(t_c)|>N_{nom}$ \textbf{do}\\
 Evaluate
\begin{equation}\label{eq_fij}
\begin{split}
   &       c_{i,j}(\bm{\alpha}_i) \\
  =   &
 \left\{
        \begin{array}{ll}
      \frac{ \min \Gamma_{j}^+(\bm{\alpha}_i) -  \max   \Gamma_{j}^-(\bm{\alpha}_i)    }{\mathcal{F}_j^{max}(\mathcal{A}(t_c))-\mathcal{F}^{min}_j(\mathcal{A}(t_c))}  , & \hbox{if } \Gamma_{j}^+(\bm{\alpha}_i), \Gamma_{j}^-(\bm{\alpha}_i) \neq \emptyset \\
          N, & \hbox{otherwise}
        \end{array}
      \right.
\end{split}
\end{equation}
for all $\bm{\alpha}_i\in \mathcal{A}(t_c) $ and $j=1,2,...,N$,
where
\begin{equation*}
    \begin{split}
 &    \mathcal{F}^{max}_j(\mathcal{A}(t_c)){ }={ } \max_{ \bm{\alpha}\in \mathcal{A}(t_c)} [\bm{\mathcal{F}}(\bm{\alpha})]_j, \\
 &    \mathcal{F}^{min}_j(\mathcal{A}(t_c) ) { }={ }  \min_{ \bm{\alpha}\in \mathcal{A}(t_c)} [\bm{\mathcal{F}}(\bm{\alpha})]_j, \\
 &    \Gamma_{j}^+(\bm{\alpha}_i) { }={ }  \{ [\bm{\mathcal{F}}(\bm{\alpha})]_j: \bm{\alpha}  \in \mathcal{A}(t_c),  [\bm{\mathcal{F}}(\bm{\alpha})]_j > [\bm{\mathcal{F}}(\bm{\alpha}_i)]_j   \}, \mbox{ and} \\
      &\Gamma_{j}^-(\bm{\alpha}_i) { }={ }  \{[\bm{\mathcal{F}}(\bm{\alpha})]_j: \bm{\alpha}  \in \mathcal{A}(t_c),  [\bm{\mathcal{F}}(\bm{\alpha})]_j < [\mathcal{F}(\bm{\alpha}_i)]_j   \}.  \\
    \end{split}
\end{equation*}
 Evaluate
\begin{equation*}
    (\bm{\alpha}_i)_{av}=  \sum_{j=1}^{N} c_{i,j}(\bm{\alpha}_i)
\end{equation*}
for all $\bm{\alpha}_i\in \mathcal{A}(t_c) $.\\
Remove the element  $\bm{\alpha}$  that yields the least $(\bm{\alpha})_{av}$  from $\mathcal{A}(t_c)$  and thus the size of $\mathcal{A}(t_c)$ is reduced by one.\\
\textbf{End While}\\
\noindent\rule{8.8cm}{0.4pt}

\section{Proposed Algorithm}\label{sec_alg}

This section presents the HMOIA used to solve~(\ref{eq_BMI_equi}).
The algorithm is a hybrid because it integrates both stochastic and deterministic search schemes.
For example,  the Levenberg--Marquardt algorithm, density reduction algorithm,
and interior-point methods are deterministic algorithms, while artificial immune systems used as the underlying structure of the HMOIA are stochastic search methods.
There are a few reasons why the immune search scheme was adopted in our main algorithm structure.
First, its potential to provide novel solutions has been illustrated in several studies~\cite{AISp1,AISp2,AISp3}.
Second, the immune search scheme is robust and outperforms some existing MOEAs or at least performs equally well in most benchmark MOPs~\cite{AIS1}.
(In~\cite{AIS1}, the MO immune algorithm was compared to PAES, PESA, NSGA-II~\cite{NSGA2}, SPEA2, MOEA/D~\cite{MOEAD2}, and ACSAMO in terms of convergence, diversity, uniformity, and coverage.) Finally and most importantly, the artificial immune system is a ``highly parallel intelligent system''~\cite{para1,para2,B_MOEA1} and thus a parallel computation scheme can be readily developed.
This is useful for solving BMI constraints that requires a large amount of computational power in general.
Despite these reasons, it is worth mentioning that other advanced MOEAs can also be adopted  if modified properly.\footnote{Proper modifications may include incorporation of pole-placement techniques into the search engine and a design of a mechanism that ensures legitimate pole placement.}

The pseudocode of the HMOIA is presented as follows.

\noindent\rule{8.8cm}{0.4pt}\\
\textbf{Pseudocode of the Proposed HMOIA}\\
\noindent\rule{8.8cm}{0.4pt}\\
\textbf{Input: }MOP in~(\ref{eq_BMI_equi})\\
Prescribe bounds on the external variable and initialize the population\\
Evaluate the objective function\\
Remove dominated points\\
Let $t_c:=1$\\
\textbf{While} $t_c\leq t_{max}$ \textbf{do}
\begin{itemize}
\item[] \textbf{If} $N=0$ and $\exists \bm{\alpha} \in \mathcal{A}(t_c)$ such that $\tilde{{\bm{\mathcal{F}}}}(\bm{\alpha})=0$\\
Let $t_c:= t_{max}$\\
       \textbf{Else}
\begin{itemize}
\item[]  Perform the hyper-mutation operation
\item[]  Evaluate the objective function
\item[]  Update population
\end{itemize}
\textbf{End If}
\item[] Let $t_c:=t_c+1$
\end{itemize}
\textbf{End While}\\
Remove points~$\bm{\alpha}$ that have $[\tilde{\bm{\mathcal{F}}}(\bm{\alpha})]_{N+1}>0$\\
Remove dominated points\\
\textbf{Output: }Approximate Pareto optimal solutions and Pareto front\\
\noindent\rule{8.8cm}{0.4pt}\\

In the following subsections,
we elaborate key steps of the algorithm and summarize the MRV.

\subsection{Prescribe Bounds on External Variable and Initialize the Population}\label{subsec_pole_ensure}

To specify the range of interests, we prescribe bounds for the external variable~$\bm{\alpha}$ in~(\ref{eq_BMI_equi}).
Entries of $\bm{\alpha}$ can be generated pointwisely over prescribed bounds
or recovered collectively from a given vector of eigenvalues~$\bm{\lambda}^{pre}$ described in~(\ref{eq_cost_L_M}) and~(\ref{eq_L_M_prob}).
For example, if a range $[\alpha^{min}_i,\alpha^{max}_i]$
is prescribed, then the $i$th entry of $\bm{\alpha}$ can be generated uniformly at random over the range.
Otherwise, if a range $[-\sigma^{min},0]\times [-\omega^{max},\omega^{max}] $ is given,
 we can randomly generate
\begin{equation}\label{eq_bounds}
 [\bm{\lambda}^{pre}]_i   \in  \{\sigma+j\omega :(\sigma,\omega ) \in [-\sigma^{min},0]\times [-\omega^{max},\omega^{max}]  \}
\end{equation}
where complex entries of $\bm{\lambda}^{pre}$  occur in conjugate pairs,
and then recover the entries of $\bm{\alpha}$ from
$ \bm{q}^*(\bm{\lambda}^{pre})$ defined in~(\ref{eq_L_M_prob}) using the trust region Levenberg--Marquardt algorithm presented in Section~\ref{subsec_L-M}.

After specifying the range, we  initialize the population:
assign the nominal population size $N_{nom}$ and the maximum population size $N_{max}$,
and generate initial population
\begin{equation}\label{eq_ini_pop}
 \{   \bm{\alpha}_1,\bm{\alpha}_2,...,\bm{\alpha}_{   N_{nom}  }  \}.
\end{equation}
The basic structure of  artificial immune algorithms in~\cite{TSG_15} and~\cite{AIS1} is adopted.
During the evolutionary process that mimics operations in immune systems, the population size changes over the iteration but remains below~$N_{max}$, nondominated points are maintained, and dominated points are removed from the population.

In general, large  bounds should be assigned to provide a spacious search space,
but such a spacious space can yield ineffective search when Pareto optimal solutions have most entries that are close to zero.
To manage possible ineffectiveness, we  divide the search spaces~$[\alpha^{min}_i,\alpha^{max}_i]$ or  $[-\sigma^{min},0]\times [-\omega^{max},\omega^{max}]  $
into several subspaces
\begin{equation}\label{eq_subspace}
     [\kappa_s \alpha^{min}_i,\kappa_s \alpha^{max}_i] \mbox{ or } [-\kappa_s \sigma^{min},0]\times [-\kappa_s \omega^{max},\kappa_s \omega^{max}]
\end{equation}
where $\kappa_s \in (0,1]$ with $s=1,2,...,S$.
Population initialization is thus modified accordingly. We either
generate $[\bm{\alpha}]_i\in [\kappa_s \alpha^{min}_i,\kappa_s \alpha^{max}_i]$ pointwisely or
recover entries of $\bm{\alpha}$ collectively from
\begin{equation}\label{eq_initialization}
\footnotesize{
 [\bm{\lambda}^{pre}]_i   \in  \{\sigma+j\omega :(\sigma,\omega ) \in [-\kappa_s \sigma^{min},0]\times [-\kappa_s \omega^{max},\kappa_s \omega^{max}] \}.
}
\end{equation}

\subsection{Evaluate the Objective Function}

 Function evaluation
for the objective function~$\tilde{\bm{\mathcal{F}}}$ defined in~(\ref{eq_F_tilde})
 can be divided into two parts.
The first part addresses the evaluation of $\bm{\mathcal{F}}(\bm{\alpha})$.
If $\bm{\mathcal{F}}$ is explicitly expressed as a function of~$\bm{\alpha}$, then  the evaluation is simply the substitution of  $\bm{\alpha}$ into $\bm{\mathcal{F}}$;
otherwise, deterministic algorithms are employed to evaluate $\bm{\mathcal{F}}(\bm{\alpha})$.
For example, if $\bm{\alpha}$ represents a controller gain of a linear control system and
$\bm{\mathcal{F}}(\bm{\alpha})$ denotes the associated $H_\infty$ norm,
then  $\bm{\mathcal{F}}(\bm{\alpha})$ must be evaluated using deterministic algorithms.
The second part addresses the evaluation of $\lambda^*(\bm{\alpha})$.
Because this evaluation is related to solving LMIs,
 deterministic algorithms such as interior-point methods can be used.

\subsection{Remove Dominated Points}

By removing dominated points from the population,  nondominated points are maintained.
Preserving nondominated points is an important operation that relates to the convergence of the algorithm.
Other operations such as the hyper-mutation and population update that guide the population towards the Pareto optimal set are important to the algorithm convergence as well.

\subsection{Perform Hyper-mutation Operation}

Let $\mathcal{A}(t_c)$ and $|\mathcal{A}(t_c)|$  denote the  current population and the associated population size, respectively.
For two vectors $\bm{a}$ and $\bm{b}$,
 $\bm{a} \oplus \bm{b}$ denotes a random and pointwise combination of
 entries of  $\bm{a}$ and $\bm{b}$, i.e.,
$[\bm{a} \oplus \bm{b}]_i$ can be either $[\bm{a}]_i$  or  $[\bm{b}]_i$ with equal probability.
For a hyper-mutation operation, new points~$\bm{\alpha}^j_i$ are generated by
\begin{equation}\label{eq_mutation_operation}
    \bm{\alpha}^j_i=
    \left\{
      \begin{array}{ll}
          L^j \bm{\alpha}_i  + (1-L^j) \bm{\alpha}^j, & \hbox{$rand>0.5$,} \\
  \bm{\alpha}_i  \oplus \bm{\alpha}^j    , & \hbox{otherwise}
      \end{array}
    \right.
 \end{equation}
for all  $\bm{\alpha}_i \in \mathcal{A}(t_c)$,
where $rand$ and $L^j$ are independent random numbers chosen from $(0,1)$, and
entries of $\bm{\alpha}^j$ are generated pointwisely or recovered collectively in the same way described at the population initialization.
The operation in~(\ref{eq_mutation_operation}) is performed $R(t_c)=\llcorner N_{max}/|\mathcal{A}(t_c)| \lrcorner$ times for each $i$, where $\llcorner \cdot \lrcorner$ represents the floor function. This operation can be interpreted as follows:
an $\bm{\alpha}_i$ in $\mathcal{A}(t_c)$ is cloned $R(t_c)$ times, and then all these cloned points  mutate to produce points $\bm{\alpha}^j_i, j=1,2,...,R(t_c)$.

\subsection{Update Population}

\begin{figure}
\centering
  \includegraphics[width=9cm]{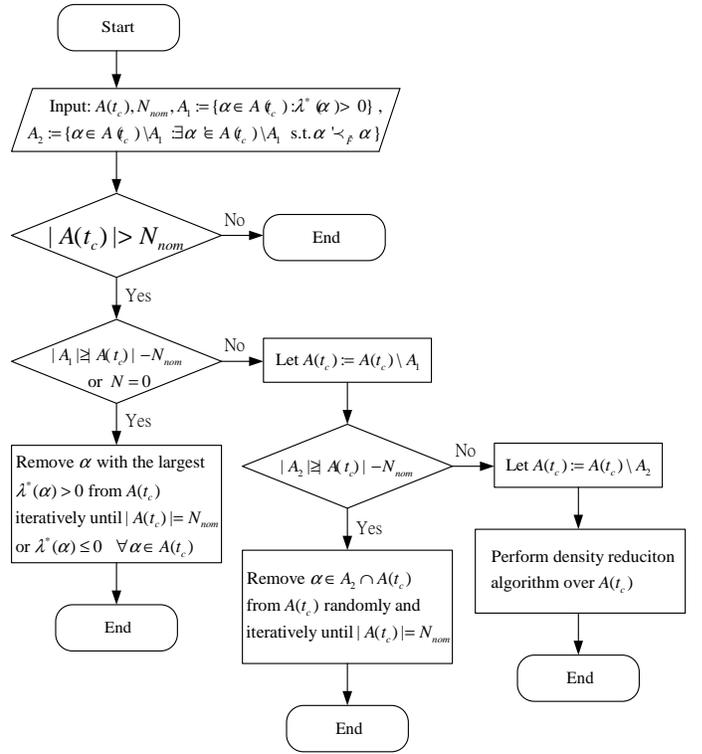}\\
  \caption{Flowchart for the removal procedure during the population update. The set $A_1$ consists of infeasible points in the current population~$\mathcal{A}(t_c)$, and $A_2$ consists of feasible but dominated points in the current population~$\mathcal{A}(t_c)$.
 The values of $\lambda^*(\bm{\alpha})$ are available
because they have been obtained during  the objective function evaluations.
  }\label{fig_update_flowchart}
\end{figure}

Updating the population consists of addition and removal operations.
After the hyper-mutation operation, $R(t_c)\times |\mathcal{A}(t_c)|$ points are newly generated and added to the population.
To keep a manageable size of the population,
we remove infeasible points, dominated points, or nondominated points in order if necessary.
Fig.~\ref{fig_update_flowchart} shows
a removal procedure that
reduces the size of $\mathcal{A}(t_c)$  to  $N_{nom}$.
In the  procedure, infeasible points $\bm{\alpha}$ that have $\lambda^*(\bm{\alpha})>0$ are gradually removed from the population.
After the removal, if the population size is still greater than its nominal size,
then dominated points are removed from the population randomly and iteratively.
If $|\mathcal{A}(t_c)|$ is still greater than $N_{nom}$, then we remove nondominated points using the density reduction algorithm described in Section~\ref{subsec_density}.

\subsection{Parameter Selection}

The values of parameters $N_{max}$, $N_{nom}$, and $t_{max}$ can affect algorithm performance.
In general, larger values of them yield a better level of performance if the complexity is not a concern~\cite{AIS1,AIS01,AIS02,AIS03}.
This is because larger values of $N_{max}$ and  $N_{nom}$ mean that more computational resources are employed to explore the search space in each iteration, and
a larger value of $t_{max}$ corresponds to more exploration time.
When the values of the parameters exceed certain thresholds,
mature convergence is attained and further improvement can be hardly observed.
Some practitioners suggest that
 large values of $N_{max}$, $N_{nom}$, and $t_{max}$ be set first, and then
these values be lowered gradually until unacceptable results are obtained.
However, this practice suffers from two drawbacks.
First, the notion of unacceptable results is vague.
Second, computational time can be a cost, and repeating the whole search with different parameters becomes costly.

Although parameters should be set differently in different problems for better performance,
we selected the same parameters in our simulations when solving all the BMI problems.
The main reason of such selection is that it is difficult to define the ``optimal'' values for parameters
 in consideration of the performance, complexity, and computational time.
Despite of using the same parameters, the simulation results still provide a proof of concept that
the proposed methodology can outperform existing design approaches
in most benchmark BMI-based design problems.

\subsection{Method of Reduction of Variables}

For a system or control design problem constrained by BMIs,
the proposed method of solution can be realized by the following steps:\\
\emph{S1)} Classify the decision variables into the internal and external variables using the Principle of Variable Classification.\\
\emph{S2)} Transform the BMI problem  in~(\ref{eq_def_BMI}) into its equivalent form in~(\ref{eq_BMI_equi}).\\
\emph{S3)} Apply the proposed HMOIA to solve~(\ref{eq_BMI_equi}).

After these three steps, the associated system or controller can  be constructed
 based on the obtained solution(s).

\begin{rmk}
There are two circumstances in which the proposed methodology can fail:
 the classification principle is not applicable; and same eigenvalues are assigned when pole placement is performed.
The first situation may occur when
 physical or mathematical bounds on coupled decision variables  cannot be obtained or readily prescribed.
Here is an example:
\begin{equation}\label{eq_not_appl}
    \bm{P}(\bm{A}+\bm{BFC})+(\bm{A}+\bm{BFC})^T\bm{P}^T<0, \bm{P}>0, \mbox{ and } \bm{F}>0
\end{equation}
where $\bm{P}$ and $\bm{F}$ are decision variables, yielding a BMI problem. In~(\ref{eq_not_appl}), we cannot assign $\bm{\alpha}=\bm{P}$ or $\bm{\alpha}=\bm{F}$ because both of them do not have inherent bounds on their entries
and  have a constraint of positive definiteness. The classification principle is thus not applicable.
Fortunately, although the proposed principle is not valid,
we rarely encounter this type of problem such as~(\ref{eq_not_appl}) in system and control designs.
When a controller design problem is considered,  $\bm{F}$  generally relates to a controller gain
and  physical constraints do not yield a requirement of positive or negative definiteness on $\bm{F}$.
In fact, $\bm{F}$ may not even be a square matrix in practice.
The second situation is related to the differentiability of eigenvalues.
Eigenvalues are differentiable only if they are distinct, and the condition of differentiability is used when pole placement is performed.
Since eigenvalues are randomly assigned in our algorithm, there is little chance that two eigenvalues are the same.
Even if their values are slightly different, the trusted region algorithm described in Section~\ref{subsec_L-M} can still work~\cite{TR_1}.
Therefore, this situation does not impose a serious restriction  on the applicability of our methodology either.
\end{rmk}

\begin{rmk}
In our framework, SOPs and MOPs are addressed in a unified manner.
For an illustrative purpose, we  examine the problem
 \begin{equation}\label{eq_modest_ex}
   \begin{array}{cc}
    \min\limits_{\bm{\alpha}}   & f(\bm{\alpha})   \\
    \mbox{ subject to } & g(\bm{\alpha}) \leq 0 \\
   \end{array}
 \end{equation}
 which can represent an SOP or MOP depending on the dimension of $f(\cdot)$.
The problem in~(\ref{eq_modest_ex}) can then be transformed into the MOP
 \begin{equation}\label{eq_modest_after}
   \begin{array}{cc}
    \min\limits_{\bm{\alpha}}   &
 \left[
   \begin{array}{cc}
    f(\bm{\alpha})  &  \max\{ g(\bm{\alpha}),0 \} \\
   \end{array}
 \right]^T.
   \end{array}
 \end{equation}
If  an SOP is considered in~(\ref{eq_modest_ex}), then the resulting problem in~(\ref{eq_modest_after}) is a 2-D MOP.
If an MOP with two objectives is considered in~(\ref{eq_modest_ex}), then the resulting problem in~(\ref{eq_modest_after}) becomes a 3-D MOP.
 For either case,~(\ref{eq_modest_after}) is regarded as an MOP and can be solved by our hybrid algorithm, producing a solution set.
 If any points in the obtained solution set  yield a nonzero value of the  final objective, i.e., $\max\{ g(\bm{\alpha}),0 \}>0$,
  then they are removed from the set because they are infeasible.
  After the removal, a legitimate APF and approximate Pareto optimal set can be attained.
The reader can refer to~\cite{SOP_MOP} for a similar technique that relates an SOP to an MOP.
\end{rmk}

\section{Numerical Examples}\label{sec_sim}

This section presents various system and control design examples using BMI approaches.
Among the solution methods included for comparison, only BB methods involve global optimization.
A detailed description of design problems and  associated system parameters can be found in the appendices.
Sections~\ref{sub_sim_feas} and~\ref{sub_sim_SO} examine
feasibility problems and SOPs constrained by BMIs, respectively.
Algorithm parameters
$N_{nom}=40$,
$N_{max}=160$,
and $t_{max}=20$ were used, and
70 simulation runs were performed.
For the BMI-based MOPs in Section~\ref{sub_sim_MO},
the iteration number $t_{max}=300$ was used to produce APFs.
These parameters were chosen based on a number of experiments in consideration of the algorithm convergence and computational time.

\subsection{Feasibility Problems with BMI Constraints}\label{sub_sim_feas}

Table~\ref{tab_feas} presents our simulation results. See Appendix~\ref{app_feas} for detailed problem descriptions.
The ``SR~$\%$'' represents the success rate of the proposed MRV solving these  feasibility problems with BMI constraints.
While the AM, ILMI, diffeomorphic state transformations, and two-step procedure were able to solve respective problems,
our method successfully found solutions in a unified manner.

\begin{table}
  \centering
  \caption{Feasibility Problems}\label{tab_feas}
  {\footnotesize
  \begin{tabular}{|c|c||c|}
    \hline
\multicolumn{1}{|c|}{ Problems}   &   Existing Solution Methods  &   \multicolumn{1}{|c|}{ Results of MRV}    \\
 Name  &       &  SR $\%$ \\
   \hline
 ST~\cite{AM} &  AM   & 100 \\
  \hline
  SIP~\cite{AM}& AM  & 100 \\
  \hline
SAFS-I~\cite{SAFS1}&  ILMI  & 100 \\
  \hline
 SAFS-II~\cite{SAFS2}& Diffeomorphic state transformations  & 100 \\
   \hline
  OCS~\cite{OCS} &  Two-step procedure  & 100   \\
    \hline
  \end{tabular}
  }
\end{table}

\begin{table*}
  \centering
  \caption{Single-objective Optimization Problems}\label{tab_SOP}
  {\footnotesize
  \begin{tabular}{|c|c|c||c|c|c|c|}
   \hline
     \multicolumn{1}{|c|}{Problems}    &     \multirow{2}{*}{Methods}  &  \multirow{2}{*}{ Results of Existing Methods } &
      \multicolumn{4}{|c|}{  Results of MRV  }\\
 Name  &   &   &  Optimum  &   Mean   &   Std    &  SR $\%$   \\
   \hline
      LPVS (maximization)~\cite{LPVS}&  BB methods & 4.75  & 4.7575  &  4.7209 &  0.0280  &  100  \\
    \hline
      SSS (maximization)~\cite{SSS}&  Path-following methods & 1.05 &  \textbf{4.1765}  &   \textbf{ 3.3347} &      0.3758   &  100  \\
         \hline
   MCD  (minimization)~\cite{ostertag2008improved}&  Path-following methods  & 0.7489 &     0.7600 &   0.8291 &    0.0407 &  100  \\
    \hline
  \end{tabular}
  }
\end{table*}

\begin{table*}
  \centering
  \caption{Spectral Abscissa Optimization}\label{tab_spectral_abs}
{\footnotesize
\begin{tabular}{ |l|r|r|r|r|r|r|r|r||r|r|r|r| }
 \hline
\multicolumn{2}{|c|}{Problems}&\multicolumn{7}{|c||}{Results of Existing Solution Methods, $\alpha_o(\bm{A}_{\bm{F}})$}&\multicolumn{4}{|c|}{Results of MRV, $\alpha_o(\bm{A}_{\bm{F}})$}\\
 Name& $\alpha_o(\bm{A})$   &HIFOO&LMIRank&PENBMI&CCDM&ICAM&   Min   & Mean   & Min & Mean &  Std & SR $\%$ \\
 \hline
 AC1&0.000 &-0.2061&-8.4766&-7.0758&-0.8535&-0.7814&  -8.4766&-3.4786    &\textbf{-18.0761}&\textbf{-11.8993}&3.2210&100\\
 \hline
 AC4&2.579   &-0.0500&-0.0500&-0.0500&-0.0500&-0.0500&  -0.05&-0.05      &-0.05&-0.05&6.9e-17&100\\
 \hline
 AC5&0.999  &-0.7746&-1.8001&-2.0438&-0.7389&-0.7389&  -2.0438&-1.2192    &\textbf{-2.4051}&\textbf{-2.1444}&0.1754&100\\
 \hline
 AC7&0.172  &-0.0322&-0.0204&0.0896&-0.0673&-0.0502&  -0.0673&-0.0161    &\textbf{-0.0747}&\textbf{-0.0494}&0.0088&100\\
 \hline
 AC8&0.012  &-0.1968&-0.4447&0.4447&-0.0755&-0.0640&  -0.4447&-0.0672    &-0.4447&\textbf{-0.4447}&2.7e-16&100\\
 \hline
 AC9&0.012  &-0.3389&-0.5230&-0.4450&-0.3256&-0.3926&  -0.523&-0.405   &\textbf{-2.0823}&\textbf{-0.5776}&0.2970&100\\
 \hline
 AC11&5.451  &-0.0003&-5.0577&x&-3.0244&-3.1573&  -5.0577&-2.8099  &\textbf{-16.9018}&\textbf{-10.6947}&2.6689&100\\
 \hline
 AC12&0.580  &-10.8645&-9.9658&-1.8757&-0.3414&-0.2948&  -10.8645&-4.6684   &\textbf{-18.3236}&\textbf{-13.3959}&2.8633&100\\
 \hline
 HE1&0.276  &-0.2457&-0.2071&-0.2468&-0.2202&-0.2134&  -0.2468&-0.2266   &-0.2446&\textbf{-0.2338}&0.0107&100\\
 \hline
 HE3&0.087  &-0.4621&-2.3009&-0.4063&-0.8702&-0.8380&  -2.3009&-0.9755   &-1.7847&-0.8908&0.3055&100\\
 \hline
 HE4&0.234  &-0.7446&-1.9221&-0.0909&-0.8647&-0.8375&  -1.9221&-0.8919    &\textbf{-3.0567}&\textbf{-1.2306}&0.5201&100\\
 \hline
 HE5&0.234  &-0.1823&x&-0.2932&-0.0587&-0.0609&  -0.2932&-0.1487   &\textbf{-1.1953}&\textbf{-0.6939}&0.2193&100\\
 \hline
 HE6&0.234  &-0.0050&-0.0050&-0.0050&-0.0050&-0.0050&  -0.005&-0.005   &-0.005&-0.005&2.6e-18&100\\
 \hline
 REA1&1.991  &-16.3918&-5.9736&-1.7984&-3.8599&-2.8932&  -16.3918&-6.1833   &\textbf{-19.3041}&\textbf{-15.4064}&2.4190&100\\
 \hline
 REA2&2.011  &-7.0152&-10.0292&-3.5928&-2.1778&-1.9514&  -10.0292&-4.9532   &\textbf{-19.4238}&\textbf{-13.0948}&4.2323&100\\
 \hline
 REA3&0.000  &-0.0207&-0.0207&-0.0207&-0.0207&-0.0207&  -0.0207&-0.0207   &-0.0207&-0.0207&3.5e-15&100\\
 \hline
 DIS2&1.675  &-6.8510&-10.1207&-8.3289&-8.4540&-8.3419&  -10.1207&-8.4193   &\textbf{-19.4340}&\textbf{-16.6852}&2.5153&100\\
 \hline
 DIS4&1.442  &-36.7203&-0.5420&-92.2842&-8.0989&-5.4467&  -92.2842&-28.6184  &-16.0222&-11.4094&2.4090&100\\
 \hline
 WEC1&0.008  &-8.9927&-8.7350&-0.9657&-0.8779&-0.8568&  -8.9927&-4.0856   &\textbf{-11.9629}&\textbf{-6.1804}&2.2291&100\\
 \hline
 IH&0.000  &-0.5000&-0.5000&-0.5000&-0.5000&-0.5000&  -0.5&-0.5  &-0.1576&-0.0617&0.0407&76.47\\
 \hline
 CSE1&0.000  &-0.4509&-0.4844&-0.4490&-0.2360&-0.2949&  -0.4844&-0.383  &-0.3489&-0.2282&0.0452&100\\
 \hline
 TF1&0.000  &x&x&-0.0618&-0.1544&-0.0704&  -0.1544&-0.0955 &\textbf{-0.2688}&\textbf{-0.1769}&0.0396&100\\
 \hline
 TF2&0.000  &x&x&-1.0e-5&-1.0e-5&  -1.0e-5 &  -1.0e-5&  -1.0e-5  &-1.0e-5&-1.0e-5&1.7e-20&100\\
 \hline
 TF3&0.000  &x&x&-0.0032&-0.0031&-0.0032&  -0.0032&-0.0031  &-0.0032&\textbf{-0.0032}&8.0e-6&100\\
 \hline
 NN1&3.606  &-3.0458&-4.4021&-4.3358&-0.8746&0.1769&  -4.4021&-2.4962  &\textbf{-5.89}&\textbf{-5.6847}&0.1812&100\\
 \hline
 NN5&0.420  &-0.0942&-0.0057&-0.0942&-0.0913&-0.0490&  -0.0942&-0.0668  &-0.094&\textbf{-0.0915}&0.0018&100\\
 \hline
 NN9&3.281  &-2.0789&-0.7048&x&-0.0279&0.0991&  -2.0789&-0.6781  &\textbf{-17.8516}&\textbf{-12.1047}&2.7270&100\\
 \hline
 NN13&1.945  &-3.2513&-4.5310&-9.0741&-3.4318&-0.2783&  -9.0741&-4.1133  &\textbf{-13.6061}&\textbf{-8.5606}&4.7341&100\\
 \hline
 NN15&0.000  &-6.9983&-11.0743&-0.0278&-0.8353&-1.0409&  -11.0743&-3.9953  &-10.9821&\textbf{-10.3002}&0.8034&100\\
 \hline
 NN17&1.170  &-0.6110&-0.5130&x&-0.6008&-0.5991&  -0.611&-0.3244  &-0.6107&\textbf{-0.6007}&0.0196&100\\
 \hline
\end{tabular}
}
\end{table*}

\begin{table*}
  \centering
  \caption{$H_2$ Optimization}\label{tab_H2}
{\footnotesize
\begin{tabular}{ |l|r|r|r|r|r|r||r|r|r|r| }
 \hline
\multicolumn{2}{|c|}{Problems}  &\multicolumn{5}{|c||}{Results of Existing Solution Methods, $||\bm{G}_{c\ell}(\bm{F})||_2$}&\multicolumn{4}{|c|}{Results of MRV, $||\bm{G}_{c\ell}(\bm{F})||_2$}\\
 Name  & $||\bm{G}_{o\ell}||_2$   &         HIFOO  &PENBMI&CCDM  & Min   & Mean  & Min & Mean &  Std & SR $\%$ \\
  \hline
 AC1& Inf&0.025&0.0061&0.054&0.0061&0.0283&		0.015&\textbf{0.0187}&0.0019&100\\
 \hline
 AC2& Inf&0.0257&0.0075&0.054&0.0075&0.029&		0.01566&\textbf{0.0188}&0.0017&100\\
 \hline
 AC3& 25.5798&2.0964&2.0823&2.1117&2.0823&2.0968&		2.1206&2.231&0.0836&100\\
 \hline
 AC4& Inf&11.0269&x&11.0269&11.0269&11.0269&	11.0269&11.0269&2.82e-15&100\\
 \hline
 AC6& 24.6067&2.8648&2.8648&2.8664&2.8648&2.8653&		3.026&3.6263&0.6533&100\\
 \hline
 AC7& Inf&0.0172&0.0162&0.0176&0.0162&0.017&	0.0162&\textbf{0.0164}&0.0001&100\\
 \hline
 AC8& Inf&0.633&0.7403&0.6395&0.633&0.6709&		\textbf{0.6321}&0.6813&0.0417&100\\
 \hline
 AC12& Inf&0.0022&0.0106&0.0992&0.0022&0.0373&		0.0627&0.1129&0.0295&100\\
 \hline
 AC15& 176.4515&1.5458&1.4811&1.5181&1.4811&1.515&		1.6564&1.828&0.1709&100\\
 \hline
 AC16& 176.4515&1.4769&1.4016&1.4427&1.4016&1.4404&		1.4641&1.5307&0.0335&100\\
 \hline
 AC17& 10.2650&1.5364&1.5347&1.5507&1.5347&1.5406&		1.5392&1.5429&0.0041&100\\
 \hline
 HE2& 13.8541&3.4362&3.4362&4.7406&3.4362&3.871&		3.7494&6.1145&1.1487&100\\
 \hline
 HE3& Inf&0.0197&0.0071&0.1596&0.0071&0.0621&		0.0333&0.1026&0.0808&100\\
 \hline
 HE4& Inf&6.6436&6.5785&7.1242&6.5785&6.7821&		15.7738&27.0193&8.7713&100\\
 \hline
 REA1& Inf&0.9442&0.9422&1.0622&0.9422&0.9828&		0.9593&0.9864&0.0176&100\\
 \hline
 REA2& Inf&1.0339&1.0229&1.1989&1.0229&1.0852&		1.0261&\textbf{1.0319}&0.0134&100\\
 \hline
 DIS1& 5.1491&0.6705&0.1174&0.7427&0.1174&0.5102&		0.51&0.7455&0.1531&100\\
 \hline
 DIS2& Inf&0.4013&0.37&0.3819&0.37&0.3844&		0.372&\textbf{0.381}&0.0128&100\\
 \hline
 DIS3& 11.6538&0.9527&0.9434&1.0322&0.9434&0.9761&		0.997&1.0623&0.0288&100\\
 \hline
 DIS4& Inf&1.0117&0.9696&1.0276&0.9696&1.0029&		1.0644&1.1091&0.0351&100\\
 \hline
 WEC1& Inf&7.394&8.1032&12.9093&7.394&9.4688&		12.1017&16.2366&1.9181&100\\
 \hline
 WEC2& 66.5622&6.7908&7.6502&12.2102&6.7908&8.8837&		13.2889&16.5298&1.3581&100\\
 \hline
 AGS& 7.0412&6.9737&6.9737&6.9838&6.9737&6.977&		7.1807&10.1753&2.4965&100\\
 \hline
 BDT1& 0.0397&0.0024&x&0.0017&0.0017&0.002&		\textbf{3.52e-05}&\textbf{5.44e-05}&1.22e-05&100\\
 \hline
 MFP& 12.6469&6.9724&6.9724&7.0354&6.9724&6.9934&		7.0556&7.6688&0.7976&100\\
 \hline
 PSM& 3.8474&0.033&0.0007&0.1753&0.0007&0.0697&		0.0217&\textbf{0.04}&0.0149&100\\
 \hline
 EB2& 4.0000&0.064&0.0084&0.1604&0.0084&0.0776&		0.0832&0.086&0.0091&100\\
 \hline
 EB3& 1.26e03&0.0732&0.0072&0.0079&0.0072&0.0294&		0.0846&0.0918&0.0141&100\\
 \hline
 TF1& Inf&0.0945&x&0.1599&0.0945&0.1272&		0.1949&0.6965&1.2484&100\\
 \hline
 TF2& Inf&11.1803&x&11.1803&11.1803&11.1803&		11.1803&11.1803&1.48e-14&100\\
 \hline
 TF3& Inf&0.1943&0.1424&0.2565&0.1424&0.1977&		0.2568&2.0745&1.6401&97.67\\
 \hline
 NN2& Inf&1.1892&1.1892&1.1892&1.1892&1.1892&	1.1892&1.1892&3.82e-06&100\\
 \hline
 NN4& 5.5634&1.8341&1.8335&1.859&1.8335&1.8422&		1.8945&1.989&0.056&100\\
 \hline
 NN8& 5.9220&1.5152&1.5117&1.5725&1.5117&1.5331&		1.5241&1.5518&0.017&100\\
 \hline
 NN11& 0.1420&0.1178&0.079&0.1263&0.079&0.1077&		0.0972&0.1137&0.0102&100\\
 \hline
 NN13& Inf&26.1012&26.1314&62.3995&26.1012&38.2107&		30.1629&\textbf{34.4666}&3.9562&100\\
 \hline
 NN14& Inf&26.1448&26.1314&62.3995&26.1314&38.2252&		29.6438&\textbf{35.6657}&7.4852&100\\
 \hline
 NN15& Inf&0.0245&x&0.021&0.021&0.0227&		\textbf{0.0034}&\textbf{0.0035}&8.82e-05&100\\
 \hline
 NN16& Inf&0.1195&0.1195&0.1195&0.1195&0.1195&		0.1208&0.2085&0.068&100\\
 \hline
 NN17& Inf&3.253&3.2404&3.3329&3.2404&3.2754&		3.2554&3.2881&0.1843&100\\
 \hline
\end{tabular}
}
\end{table*}

\begin{table*}
  \centering
  \caption{$H_\infty$ Optimization}\label{tab_infty}
{\footnotesize
\begin{tabular}{ |l|r|r|r|r|r|r||r|r|r|r| }
 \hline
 \multicolumn{2}{|c|}{Problems}  &\multicolumn{5}{|c||}{Results of Existing Solution Methods, $||\bm{G}_{c\ell}(\bm{F})||_\infty$}&\multicolumn{4}{|c|}{Results of MRV, $||\bm{G}_{c\ell}(\bm{F})||_\infty$}\\
 Name & $||\bm{G}_{o\ell}||_\infty$   &   HIFOO  &PENBMI&CCDM  & Min   & Mean  & Min & Mean &  Std & SR $\%$ \\
 \hline
   AC1& 2.1672&0.0000&x&0.0177&0.0000&0.0088&		0.0405&0.0907&0.0285&100\\
 \hline
  AC2& 2.1672&0.1115&x&0.1140&0.1115&0.1127&		0.1262&0.1917&0.0310&100\\
 \hline
  AC3& 352.6869&4.7021&x&3.4859&3.4859&4.094&		3.9206&4.5709&0.4217&100\\
 \hline
  AC4 & 69.9900&0.9355&x&69.9900&0.9355&35.4627&		69.99&69.99&1.29e-13&100\\
 \hline
  AC6& 391.7820&4.1140&x&4.1954&4.114&4.1547&		4.8138&6.9232&3.1554&100\\
 \hline
  AC7& 0.0424&0.0651&0.3810&0.0548&0.0548&0.1669&		\textbf{0.0315}&\textbf{0.0316}&6.27e-05&100\\
 \hline
  AC8& 1.7e03&2.0050&x&3.0520&2.005&2.5285&		\textbf{1.4305}&\textbf{1.8223}&0.4017&100\\
 \hline
  AC9& Inf&1.0048&x&0.9237&0.9237&0.9642&		3.2926&5.1355&1.069&100\\
 \hline
  AC11& Inf&3.5603&x&3.0104&3.0104&3.28535&		3.1158&4.0119&0.5472&100\\
 \hline
  AC12& 586.9176&0.3160&x&2.3025&0.316&1.3092&		1.3532&1.9379&0.1729&100\\
 \hline
  AC15& 2.4e03&15.2074&427.4106&15.1995&15.1995&152.6058&		17.1925&\textbf{18.2818}&0.4480&100\\
 \hline
  AC16& 2.4e03&15.4969&x&14.9881&14.9881&15.2425&		15.8600&16.6389&0.5547&100\\
 \hline
  AC17& 30.8328&6.6124&x&6.6373&6.6124&6.6248&		6.6124&\textbf{6.6124}&1.30e-06&100\\
 \hline
  HE1& 0.5598&0.1540&1.5258&0.1807&0.154&0.6201&		\textbf{0.1538}&\textbf{0.1595}&0.0045&100\\
 \hline
  HE2& 81.8318&4.4931&x&6.7846&4.4931&5.6388&		\textbf{4.3681}&\textbf{5.5034}&0.7626&100\\
 \hline
  HE3& 1.4618&0.8545&1.6843&0.9243&0.8545&1.1543&		0.8570&\textbf{0.9142}&0.0381&100\\
 \hline
  HE4& 174.2975&23.3448&x&22.8713&22.8713&23.108&		46.5677&65.3844&7.8214&100\\
 \hline
  HE5& 2.0802&8.8952&x&37.3906&8.8952&23.1429&		20.8784&137.7817&155.8509&100\\
 \hline
  REA1& 25.7708&0.8975&x&0.8815&0.8815&0.8895&		0.8836&0.9073&0.0309&100\\
 \hline
  REA2& 26.3449&1.1881&x&1.4188&1.1881&1.3034&		\textbf{1.1471}&\textbf{1.168}&0.0125&100\\
 \hline
  REA3& Inf&74.2513&74.446&74.5478&74.2513&74.415& 74.2513&75.5692&2.3953&100\\
 \hline
  DIS1& 17.3209&4.1716&x&4.1943&4.1716&4.1829&		4.3197&4.7678&0.4376&100\\
 \hline
  DIS2& 0.9016 &1.0548&1.7423&1.1546&1.0548&1.3172&		1.0604&\textbf{1.1364}&0.039&100\\
 \hline
  DIS3& 32.0698&1.0816&x&1.1382&1.0816&1.1099&		1.2727&1.3733&0.0436&100\\
 \hline
  DIS4& 3.1304&0.7465&x&0.7498&0.7465&0.7481&		0.9486&1.0203&0.0411&100\\
 \hline
  TG1& 130.3418&12.8462&x&12.9336&12.8462&12.8899&		14.2157&25.1589&13.8431&100\\
 \hline
  AGS& 8.1820&8.1732&188.0315&8.1732&8.1732&68.126&		10.0239&\textbf{20.963}&6.2165&100\\
 \hline
  WEC2& 354.3162&4.2726&32.9935&6.6082&4.2726&14.6247&		7.8382&\textbf{10.3568}&1.5282&100\\
 \hline
  WEC3&  180.0408&4.4497&200.1467&6.8402&4.4497&70.4788&		7.2021&\textbf{9.6854}&1.4185&100\\
 \hline
  BDT1& 5.1426&0.2664&x&0.8562&0.2664&0.5613&		\textbf{0.2662}&\textbf{0.2669}&0.0008&100\\
 \hline
  MFP& 83.1407 &31.5899&x&31.6079&31.5899&31.5989&		33.9193&51.8236&25.4745&100\\
 \hline
  IH& Inf&1.9797&x&1.1858&1.1858&1.5827&		30.1004&450.2228&1231.8531&90.38\\
 \hline
  CSE1& 1.3e13 &0.0201&x&0.0220&0.0201&0.021&		\textbf{0.0198}&\textbf{0.0199}&2.00e-05&100\\
 \hline
  PSM& 4.2328&0.9202&x&0.9227&0.9202&0.9214&	0.9202&\textbf{0.9208}&0.001&100\\
 \hline
  EB1& 39.9526&3.1225&39.9526&2.0276&2.0276&15.0342&		\textbf{1.888}&\textbf{1.888}&4.81e-08&100\\
 \hline
  EB2& 39.9526 &2.0201&39.9547&0.8148&0.8148&14.2632&		\textbf{0.8142}&\textbf{0.8142}&8.44e-16&100\\
 \hline
  EB3& 3.9e06&2.0575&3995311.074&0.8153&0.8153&1331771.316&		\textbf{0.8143}&\textbf{0.8143}&6.17e-16&100\\
 \hline
  NN1& Inf &13.9782&14.6882&18.4813&13.9782&15.7159&		15.5294&16.6317&0.8991&100\\
 \hline
  NN2& Inf&2.2216&x&2.2216&2.2216&2.2216&		\textbf{2.2038}&\textbf{2.2056}&0.0021&100\\
 \hline
  NN4& 31.0435&1.3627&x&1.3802&1.3627&1.3714&		1.4327&1.6037&0.0855&100\\
 \hline
  NN8& 46.5086&2.8871&78281181.15&2.9345&2.8871&26093728.99&		2.9193&\textbf{2.9977}&0.047&100\\
 \hline
  NN9& 3.7675&28.9083&x&32.1222&28.9083&30.5152&		30.7173&35.299&7.047&100\\
 \hline
  NN11& 0.1703&0.1037&x&0.1566&0.1037&0.1301&		0.1075&0.1374&0.0127&100\\
 \hline
  NN15& Inf&0.1039&x&0.1194&0.1039&0.1116&		\textbf{0.098}&\textbf{0.0982}&0.0001&100\\
 \hline
  NN16& 6.4e14&0.9557&x&0.9656&0.9557&0.9606&		2.3044&6.8293&3.2719&100\\
 \hline
  NN17& 2.8284&11.2182&x&11.2381&11.2182&11.2281&		\textbf{11.2042}&11.6262&0.5366&100\\
 \hline
\end{tabular}
}
\end{table*}

\subsection{SOPs  with BMI Constraints}\label{sub_sim_SO}

We compared the MRV with BB and path-following methods.
Table~\ref{tab_SOP} shows the numerical results in which
the ``Optimum, Mean,'' and ``Std'' stand for the achieved optimal value, mean value, and standard deviation, respectively.
In comparison with existing methods, the MRV yielded prominent improvement in problem SSS, and had similar results in problems LPVS and MCD.

To further assess the performance of the MRV, we used various models in COMPl$_e$ib~\cite{COMPleib,COMPleib_soft}, including
aircraft models (AC), helicopter models (HE), reactor models (REA), decentralized interconnected systems (DIS),
wind energy conversion models (WEC),  terrain following models (TF), and academic test problems (NN).
 Spectral abscissa optimization problems, $H_2$ optimization problems, and
$H_\infty$ optimization problems were investigated.
The associated system under investigation has the following form:
\begin{equation}\label{eq_syst_inv}
\left\{
  \begin{array}{l}
 \bm{\dot{x}}=\bm{Ax}+\bm{B}_1\bm{w}+\bm{Bu} \\
 \bm{z}=\bm{C}_1\bm{x}+\bm{D}_{11}\bm{w}+\bm{D}_{12}\bm{u} \\
  \bm{y}=\bm{Cx }.
  \end{array}
\right.
\end{equation}
The closed-loop system of~(\ref{eq_syst_inv}) using  a static output feedback controller  $\bm{u}=\bm{Fy}=\bm{FCx }$ can be written as
\begin{equation}\label{eq_syst_cl}
\left\{
  \begin{array}{l}
 \bm{\dot{x}}=(\bm{A}+\bm{BFC})\bm{x} +\bm{B}_1\bm{w}= \bm{A}_{\bm{F}}\bm{x} +\bm{B}_1\bm{w}\\
 \bm{z}=(\bm{C}_1+\bm{D}_{12}\bm{FC})\bm{x}+\bm{D}_{11}\bm{w}=\bm{C}_{\bm{F}}\bm{x}+\bm{D}_{11}\bm{w}. \\
  \end{array}
\right.
\end{equation}

The spectral abscissa optimization (or minimization) associated with~(\ref{eq_syst_cl}) is formulated as~\cite{burke2002two}
\begin{equation}\label{eq_ex_spectrum}
   \min_{\bm{F}} \;  \alpha_o(\bm{A}_{\bm{F}})
\end{equation}
 where
\begin{equation*}
\alpha_o(\bm{A}_{\bm{F}})=\max_{\lambda\in eig\{\bm{A}_{\bm{F}} \}}   Re(\lambda)
\end{equation*}
is the spectral abscissa of $\bm{A}_{\bm{F}}$,
 $eig\{\bm{A}_{\bm{F}} \} $ represents the set of  eigenvalues of~$\bm{A}_{\bm{F}}$, $Re(\lambda)$ is the real part of $\lambda$, and
 matrix $\bm{F}$ represents the controller gain that must be determined.
Because the objective function in~(\ref{eq_ex_spectrum}) is neither smooth nor Lipschitz continuous,
(\ref{eq_ex_spectrum}) is conventionally transformed into the BMI problem~\cite{burke2002two,COMPleib}:
\begin{equation}\label{eq_BMI_spectrum}
 \begin{split}
    \min_{\bm{P},\bm{F},\beta  } \; & \beta  \\
  \mbox{subject to }&   (\bm{PA}_{\bm{F}},\star)+2\beta \bm{P}<0, \bm{P}>0
 \end{split}
\end{equation}
where $\beta$ is related to the decay rate of the system.

To use our methodology,  we compared~(\ref{eq_ex_spectrum}) to~(\ref{eq_def_BMI}),
and  let $\bm{\alpha}=\bm{F}$, $\bm{X}=\emptyset$ (no internal variable is involved), and
$\bm{\mathcal{F}}(\bm{\alpha} )=\alpha_o(\bm{A}_{\bm{F}})  $.
 Table~\ref{tab_spectral_abs} presents the resulting performance of various solution methods.\footnote{In Tables~\ref{tab_spectral_abs}--\ref{tab_infty}, the numerical results of existing methods
 HIFOO, LMIRank, PENBMI, and CCDM  come from~\cite{dinh2012combining}, and the results of ICAM are from~\cite{dinh2012inner}. The notation $\alpha_o(\bm{A})$ represents  the spectral abscissa of $\bm{A}$.}
Minimization problems are considered. Values in the columns of  Table~\ref{tab_spectral_abs} labeled with ``Min''  and ``Mean'' present the best possible performance and average performance during the simulation trials, respectively. The letter ``x'' means that no solution is found.
Our approach performed excellently in approximately $73\%$  of test problems (marked in bold numbers) and yielded similar levels of performance  in the remaining problems
as compared with existing solution methods.

\begin{rmk}
In Tables~\ref{tab_spectral_abs}--\ref{tab_add_SAO},
the ``Min''  and ``Mean''  serve as performance metrics in different situations.
Having the minimum ``Min'' in the results of the MRV implies that the proposed method outperforms existing solution methods in the best-case scenario.
The best-case scenario can be related to the situation in which the computational complexity is not a concern.
 The best solution can then be obtained by a series of evaluations.
This situation occurs in certain off-line applications and the value of ``Min'' can serve as a performance metric.
By contrast, having the minimum ``Mean'' in the results of the MRV indicates that the proposed method is better than existing solution methods in average.
When  computational resources are limited, e.g., in certain online applications, the value of ``Mean'' can serve as a performance metric.
\end{rmk}

For $H_2$ and $H_\infty$ optimization, we use
\begin{equation}\label{eq_cloop}
G_{o\ell}=
\left[
  \begin{array}{c|c}
    \bm{A} & \bm{B}_{1} \\
    \hline
   \bm{C}_{1} &\bm{D}_{11} \\
  \end{array}
\right]
\mbox{ and }
G_{c\ell}(\bm{F})=
\left[
  \begin{array}{c|c}
    \bm{A}_{\bm{F}} & \bm{B}_{1} \\
    \hline
   \bm{C}_{\bm{F}} &\bm{D}_{11} \\
  \end{array}
\right]
\end{equation}
 to represent the open- and closed-loop systems, respectively.
The controller gain~$\bm{F}$ was designed so that the $H_2$ norm of the closed-loop system, denoted by $||G_{c\ell}(\bm{F})||_2$, or
the $H_\infty$ norm of the closed-loop system, denoted by $||G_{c\ell}(\bm{F})||_\infty$,
was minimized while certain BMI constraints were satisfied (see Appendix~\ref{app_SO}).
Tables~\ref{tab_H2} and~\ref{tab_infty} present the respective results.\footnote{In the tables, $||\bm{G}_{o\ell}||_2$ and $||\bm{G}_{o\ell}||_\infty$ represent the $H_2$ and $H_\infty$ norms of the open-loop system~$\bm{G}_{o\ell}$, respectively. The notation ``Inf''  stands for ``infinity.''}
The MRV outperformed  existing solution methods
in approximately $27.5\%$ and $47.8\%$ of test problems for $H_2$ and $H_\infty$ optimization, respectively (marked in bold numbers).
For the remaining problems,
it yielded  similar levels of performance.

\begin{table}
  \centering
  \caption{Additional Problems in Spectral Abscissa, $H_2$, and $H_\infty$ Optimization}\label{tab_add_SAO}
 {\footnotesize
\begin{tabular}{ |l|c||r|r|r|r| }
 \hline
\multicolumn{2}{|c||}{Problems}&\multicolumn{4}{|c|}{Results of MRV, $\alpha_o(\bm{A}_{\bm{F}})$}\\
 Name& $\alpha_o(\bm{A})$& Min & Mean &  Std & SR $\%$ \\
 \hline
AC18  & 0.1015 &   -1.9248&  -1.1526   & 0.3836&100\\
 \hline
 DIS5 & 1.0192 &  -2.7044&   -2.3709      & 0.2454&100\\
  \hline
PAS & 0 &  -2.05e-05&      -1.61e-05         &5.87e-06&100\\
 \hline
NN12 &  1.0000  &   -2.4761&   -1.9860 &0.5362&100\\
 \hline
  \hline
\multicolumn{2}{|c||}{Problems}&\multicolumn{4}{|c|}{Results of MRV, $||\bm{G}_{c\ell}(\bm{F})||_2$}\\
 Name& $||\bm{G}_{o\ell}||_2$ & Min & Mean &  Std & SR $\%$ \\
  \hline
 AC18  & Inf&    20.0248 &    21.1601 &   0.6287 & 100 \\
  \hline
 DIS5   & Inf&  0.0013 &    0.0019&  0.0006 & 100\\
  \hline
 NN12 & Inf&      8.6989   &  10.5373 & 3.7473  & 100\\
  \hline
    \hline
\multicolumn{2}{|c||}{Problems}&\multicolumn{4}{|c|}{Results of MRV, $||\bm{G}_{c\ell}(\bm{F})||_\infty$}\\
 Name& $||\bm{G}_{o\ell}||_\infty $& Min & Mean &  Std & SR $\%$ \\
  \hline
 AC18  & 140.3365&    10.8088& 11.9210&  1.7586 & 100\\
  \hline
 DIS5   & 0.0108&  28.7928 &    29.3512&  0.4921& 100\\
  \hline
 NN12 & Inf&  22.4556 &   40.5618 &     28.6066& 100\\
  \hline
\end{tabular}
}
\end{table}

As testified in~\cite{dinh2012combining} and~\cite{dinh2012inner},
 the CCDM and ICAM were relatively robust compared with other existing solution methods, but
they failed or made little progress towards a local solution in problems AC18, DIS5, PAS, and NN12 in COMPl$_e$ib.
By contrast, the MRV was able to find solutions in these problems, as shown in Table~\ref{tab_add_SAO}.

\subsection{MOPs  with BMI Constraints}\label{sub_sim_MO}

This subsection examines the ability of the MRV to produce APFs for controller designs involving multiple objectives.
The first problem is the sparse linear constant output-feedback design described in~(\ref{eq_sparse_mod})
in which the matrices $\bm{A},\bm{B},$ and $\bm{C}$ are defined in~\cite{SSS} and~\cite{dinh2012combining}.
The goal is to maximize the decay rate $\beta$ and minimize the entry values of the controller gain~$\bm{F}$.

The second problem is an MO version of a mixed $H_2/H_\infty$ control problem (derived from (35) in~\cite{dinh2012combining}):
\begin{equation}\label{eq_ex_MO2_pre}
 \begin{split}
   \min_{\bm{P}_1,\bm{P}_2,\bm{F},\bm{Z},\gamma  } \; &
\left[
  \begin{array}{cc}
    trace(\bm{Z}) &   \gamma \\
  \end{array}
\right]^T  \\
   \mbox{subject to } &
\left[
  \begin{array}{cc}
 (\bm{P}_1 \bm{A}_{ \bm{F} },\star)+(\bm{C}_{ \bm{F} }^{\bm{z}_1})^T\bm{C}_{ \bm{F} }^{\bm{z}_1}  & \bm{P}_1 \bm{B}_1 \\
    \star & -\gamma^2 \bm{I} \\
  \end{array}
\right]<0, \\
&
\left[
  \begin{array}{cc}
 (\bm{P}_2 \bm{A}_{ \bm{F} },\star)  & \bm{P}_2 \bm{B}_1 \\
    \star & - \bm{I} \\
  \end{array}
\right]<0,
\left[
  \begin{array}{cc}
 \bm{P}_2   & \star \\
   \bm{C}_{ \bm{F} }^{\bm{z}_2} &  \bm{Z} \\
  \end{array}
\right]>0,\\
&
\bm{P}_1
, \bm{P}_2>0
 \end{split}
\end{equation}
where $\bm{A}_{ \bm{F} }=\bm{A}+\bm{BFC}$, $\bm{C}_{ \bm{F} }^{\bm{z}_i}=\bm{C}^{\bm{z}_i}+\bm{FC}, i=1,2$,
and the matrices $\bm{A},\bm{B},\bm{B}_1,\bm{C}_{ \bm{F} }^{\bm{z}_1}$ and $\bm{C}_{ \bm{F} }^{\bm{z}_2}$
are defined in~\cite{SSS} and~\cite{dinh2012combining}.
The $H_2$ and $H_\infty$ performance are related to $trace(\bm{Z})$ and $\gamma$, respectively.
To apply the classification principle, we introduced a slack variable $\eta$ and imposed an additional constraint $ trace(\bm{Z}) \leq \eta^2$.
The mixed $H_2/H_\infty$ design problem in~(\ref{eq_ex_MO2_pre}) can then be transformed into
\begin{equation}\label{eq_ex_MO2}
{\small
 \begin{split}
   \min_{\bm{P}_1,\bm{P}_2,\bm{F},\bm{Z},\eta,\gamma  } \; &
\left[
  \begin{array}{cc}
    \eta &   \gamma \\
  \end{array}
\right]^T  \\
   \mbox{subject to } &
\left[
  \begin{array}{cc}
 (\bm{P}_1 \bm{A}_{ \bm{F} },\star)+(\bm{C}_{ \bm{F} }^{\bm{z}_1})^T\bm{C}_{ \bm{F} }^{\bm{z}_1}  & \bm{P}_1 \bm{B}_1 \\
    \star & -\gamma^2 \bm{I} \\
  \end{array}
\right]<0, \\
&
\left[
  \begin{array}{cc}
 (\bm{P}_2 \bm{A}_{ \bm{F} },\star)  & \bm{P}_2 \bm{B}_1 \\
    \star & - \bm{I} \\
  \end{array}
\right]<0,
\left[
  \begin{array}{cc}
 \bm{P}_2   & \star \\
   \bm{C}_{ \bm{F} }^{\bm{z}_2} &  \bm{Z} \\
  \end{array}
\right]>0,\\
&
\bm{P}_1
, \bm{P}_2>0, trace(\bm{Z}) \leq \eta^2.
 \end{split}
 }
\end{equation}
According to the classification principle, $\bm{P}_1,\bm{P}_2$, and $\bm{Z}$ must be included in the internal variable $\bm{X}$ because they have a constraint on positive definiteness ($\bm{Z}$ is located in a diagonal block of a positive-definite matrix and hence, it is positive-definite); since $\eta$ and $\gamma$ appear in the objective function, they are included in the external variable $\bm{\alpha}$; and finally, $\bm{F}$ must be included in  $\bm{\alpha}$ as well so that
for a fixed $\bm{\alpha}$, $\mathcal{BMI}( \bm{\alpha},\bm{X} ) < 0$   becomes an LMI in the variable $\bm{X}$.
Referring to the notations in~(\ref{eq_def_BMI}), we let~$\bm{\alpha}=(\eta,\gamma,\bm{F})$, $\bm{X}=(\bm{P}_1,\bm{P}_2,\bm{Z})$, and
$\bm{\mathcal{F}}(\bm{\alpha} )=[   \eta  \; \gamma]^T$.

\begin{figure*}
  \centering
\begin{equation*}
   \begin{array}{cc}
     \includegraphics[width=8cm]{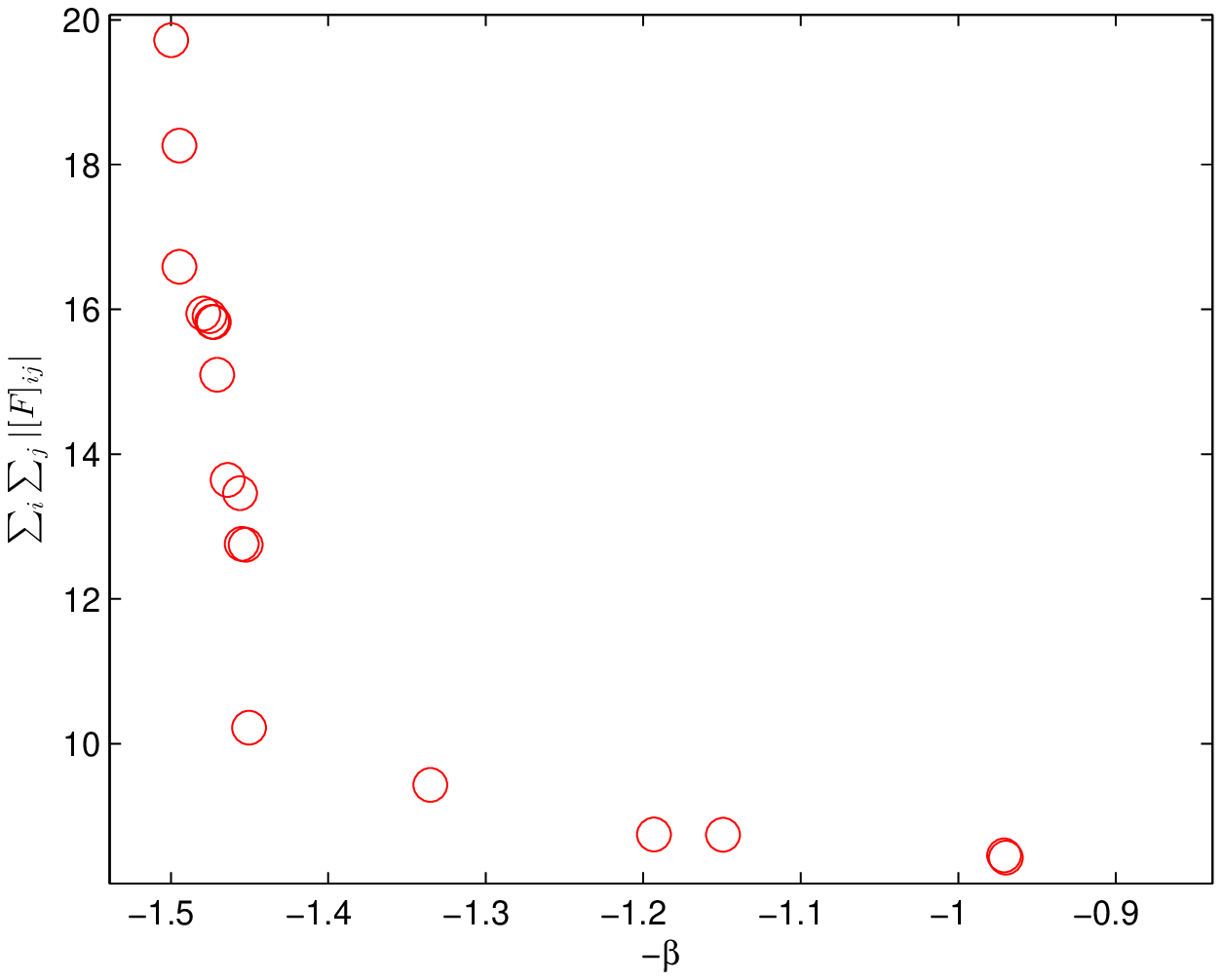} &  \includegraphics[width=8cm]{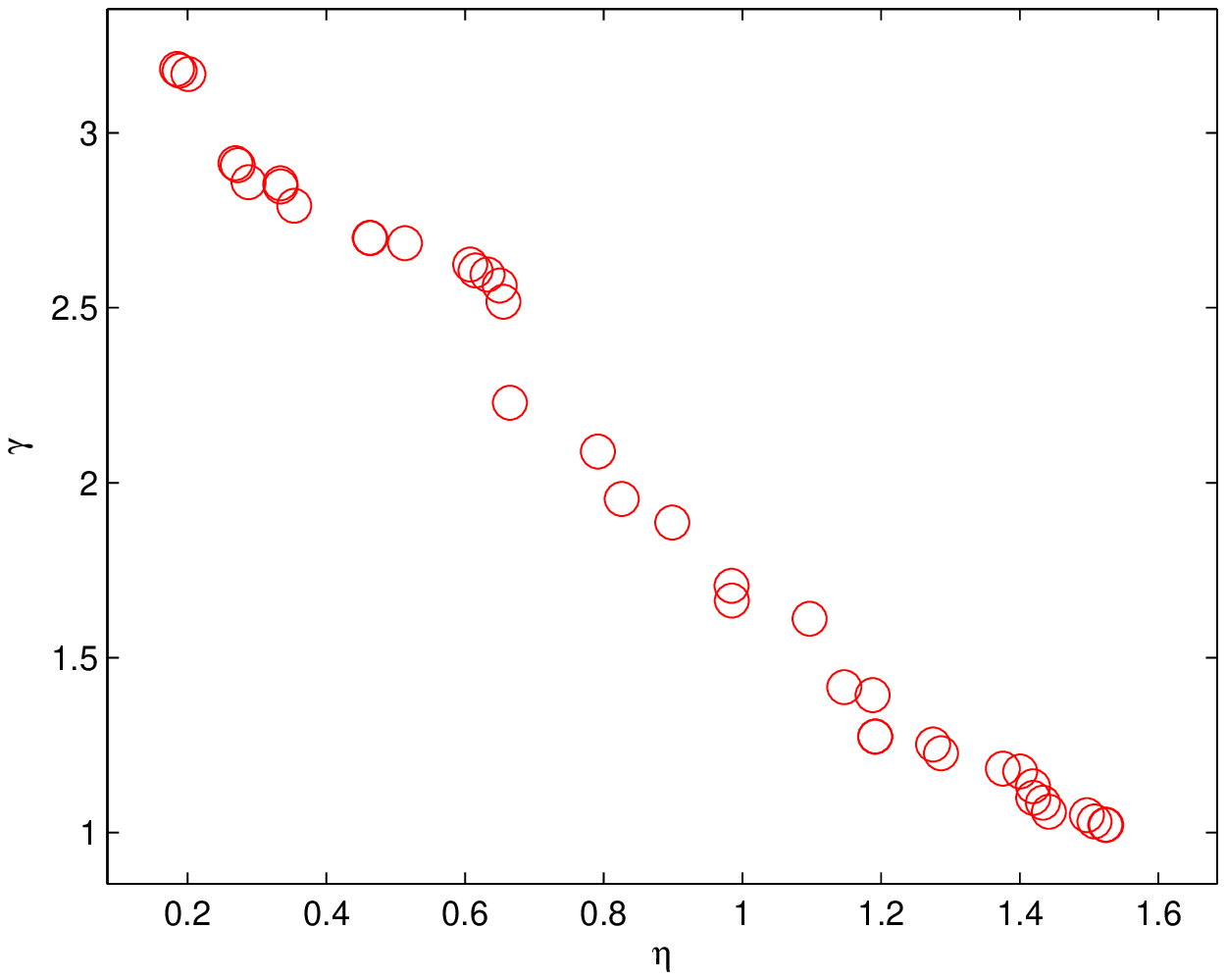} \\
     \mbox{(a)} & \mbox{(b)}
     \end{array}
\end{equation*}
  \caption{APFs obtained by solving MOPs with BMI constraints using the proposed MRV. (a) APF obtained by solving~(\ref{eq_sparse_mod}); (b) APF obtained by solving~(\ref{eq_ex_MO2}).  }\label{fig_APFs}
\end{figure*}

Fig.~\ref{fig_APFs} shows APFs obtained by solving~(\ref{eq_sparse_mod}) and~(\ref{eq_ex_MO2}).
The APF in Fig.~\ref{fig_APFs}(a) is bent, implying that objectives are not heavily dependent.
By choosing a design that corresponds to a vector in the knee region of the APF, it is possible to simultaneously improve both objectives, i.e., maximizing the decay rate and minimizing the values of entries of the controller gain.
By contrast,  the line shape of the APF in Fig.~\ref{fig_APFs}(b) indicates that  simultaneous improvement in both objectives cannot be attained.
Therefore, for the design problem in~(\ref{eq_ex_MO2}), we must sacrifice the $H_\infty$ performance to improve the $H_2$ performance, or vice versa.
These two examples illustrate  the ability of the MRV  to produce APFs for  MOPs constrained by BMIs.
Since an APF can contain useful information about the relationships among objectives,
applying the proposed methodology to solve BMI-constrained MOPs can be advantageous.

\begin{rmk}
We illustrated how beneficial it can be by using
the proposed methodology, consisting of variable classification, problem transformation, and algorithm integration,
 to solve BMI-constrained problems in system and control designs.
 The comparisons with existing BMI solution methods validated the effectiveness of the proposed methodology.
A proof of concept was thus provided.
 It is worth noting that existing MOEAs cannot be directly applied to obtain  solutions to our BMI-constrained problems such as~(\ref{eq_sparse_mod}) and~(\ref{eq_ex_MO2})
 because they are not expressed in a standard form of an MOP. (The problem in~(\ref{eq_BMI_equi}) derived from the Equivalence Theorem is in the standard form.)
Even if our variable classification and problem transformation have been performed so that the standard form has been obtained,
there is little chance that existing MOEAs can solve the resulting MOPs
because the external variable $\bm{\alpha}$ often contains the controller gain $\bm{F}$ that is still related to some matrix constraint.
The proposed hybrid algorithm can have the power to produce solutions mainly because
we have incorporated a pole-placement technique into the search engine (described in Section~\ref{subsec_L-M}) and
 designed a mechanism that ensures legitimate pole placement (described in Section~\ref{subsec_pole_ensure}).
\end{rmk}

\section{Conclusion}\label{sec_con}

In this paper, we proposed a solution method termed the MRV for BMI problems in system and control designs.
By using this method,  the associated decision variables are classified into external and internal variables
according to the variable classification principle; the BMI problem is then transformed into an unconstrained optimization problem that has fewer decision variables;
and finally, a hybrid algorithm termed HMOIA is applied to solve the unconstrained problem, yielding a feasible point, a solution, or a set of approximate Pareto optimal solutions
depending on the dimension of the objective function.
In our simulations, we compared the proposed MRV to various BMI solution methods and found that
the MRV yielded excellent levels of performance in many benchmark problems, validating the proposed methodology.
In contrast with some existing BMI solution methods, the MRV possesses the following advantages:
it expresses decision variables in a vector form, which is convenient for controller designs;
it avoids much effort such as problem reformulation or prior derivations, which can be heuristic and cumbersome;
it performs global optimization instead of local search, which is essential because BMI problems are non-convex and have multiple local optima;
and it can address multiple objectives simultaneously.

\appendices

\section{Feasibility Problems with BMI Constraints}\label{app_feas}

If not specified, the following equation numbers are those in the respective references.

ST (Stability Test, Sec. V-A of~\cite{AM} with $\mu=0.1$):
The problem was described in~(\ref{eq_ST_BMI}) of this paper,
and the bounds
$\tau_{\ell i j}\in [0,10]$ were used.

SIP  (Stabilization of Inverted Pendulum, Sec.~V-B of~\cite{AM} with $\mu=0.001$ and BMIs in~(13)):
\begin{equation*}
\begin{split}
   &
 \left[
      \begin{array}{ccc}
       ( \bm{P}_i (\bm{A}_\ell+\bm{B}_\ell \bm{F}_i  ),\star)+\mu^2 \bm{I} & \star & \star \\
        \bm{P}_i & -\bm{I} & 0 \\
       \mu \bm{F}_i & \star & -\bm{I} \\
      \end{array}
    \right]
   \\
   &
 -  \sum_{j=1}^2 \tau_{\ell i j}
    \left[
      \begin{array}{cc}
    \bm{P}_j-   \bm{P}_i & \star \\
        \bm{0} & \bm{0} \\
      \end{array}
    \right]<0 \mbox{ for } \ell,i=1,2
\end{split}
\end{equation*}
where $\tau_{\ell i j}\geq 0$, $\bm{F}_i$, and $\bm{P}_i>0$ are the decision variables.
The external and internal variables were classified as $\bm{\alpha}=(\tau_{112},\tau_{121},\tau_{212},\tau_{221},\bm{F}_1,\bm{F}_2)$ and $\bm{X}=(\bm{P}_1,\bm{P}_2)$, respectively.
The bounds $\tau_{\ell i j}\in [0,10]$ and   $[\bm{F}_i]_{mn}\in [-10,10]$ were used.

SAFS-I  (Stabilization of an Affine Fuzzy System, Sec.~V of~\cite{SAFS1} with BMIs in (16.1) and~(16.2)):
\begin{equation*}
    \bm{G}_{22}^T \bm{P} \bm{G}_{22} -\bm{P}<0
\end{equation*}
and
\begin{equation*}
\left[
  \begin{array}{cc}
\bm{G}_{ij}^T \bm{P} \bm{G}_{ij} -\bm{P}-  \tau_{ij}\bm{T}_{ij}   & \bm{G}_{ij}^T \bm{P} \bm{\sigma}_{ij} -  \tau_{ij}\bm{u}_{ij}   \\
    \star & \bm{\sigma}_{ij}^T \bm{P} \bm{\sigma}_{ij} -  \tau_{ij}v_{ij} \\
  \end{array}
\right]<0
\end{equation*}
 for   $(i,j)=(1,1),(3,3), (1,2), (2,3)$,
 where
\begin{equation*}
\bm{G}_{ij}{ }={ }  \frac{1}{2}\{ ( \bm{A}_{i}- \bm{B}_{i} \bm{F}_{j})+ ( \bm{A}_{j}- \bm{B}_{j} \bm{F}_{i})  \} \mbox{ and }  \bm{\sigma}_{ij}=\frac{1}{2} ( \bm{\mu}_i+ \bm{\mu}_j ).
\end{equation*}
The $\tau_{ i j}\geq 0$, $\bm{F}_i$, and $\bm{P}>0$ are the decision variables.
The external and internal variables were classified as $\bm{\alpha}=(\tau_{11},\tau_{33},\tau_{12},\tau_{23},\bm{F}_1,\bm{F}_2,\bm{F}_3)$ and $\bm{X}=\bm{P}$, respectively.
The bounds
   $ \tau_{ i j}\in [0,5]$ and  $[\bm{F}_i]_{mn}\in [-5,5]$
were used.

SAFS-II (Stabilization of an Affine Fuzzy System, Sec.~V of ~\cite{SAFS2} with BMIs in~(14a) and~(14b)):
\begin{equation*}
 (\bm{P} (\bm{A}_2- \bm{BF}_2),\star)<0
\end{equation*}
and
\begin{equation*}
\small{
\left[
  \begin{array}{cc}
 (\bm{P} (\bm{A}_i- \bm{BF}_i),\star)-  \tau_{ij}\bm{T}_{ij}   &  \bm{P} (\bm{\mu}_i -\bm{B}\sigma_i ) -     \tau_{ij}\bm{u}_{ij}   \\
    \star &  -  \tau_{ij}v_{ij} \\
  \end{array}
\right]<0
}
\end{equation*}
for   $(i,j)=(1,1),(3,1)$.
The $\tau_{ i j}\geq 0$, $\bm{F}_i$, and $\bm{P}>0$ are the decision variables.
The external and internal variables were classified as $\bm{\alpha}=(\tau_{11},\tau_{31},\sigma_1,\sigma_3,\bm{F}_1,\bm{F}_2,\bm{F}_3)$ and $\bm{X}=\bm{P}$, respectively.
The bounds
   $ \tau_{i j}\in [0,5],  \sigma_i\in [-5,5] $ and  $ [\bm{F}_i]_{mn}\in [-5,5]$
were used.

OCS  (Observer-based Control System, Sec.~IV-A of~\cite{OCS} with BMIs in (16)):
\begin{equation*}
\footnotesize{
    \left[
      \begin{array}{cccc}
       ( \bm{P}_1 \bm{A}_i-\bm{P}_1\bm{B}_{2i} \bm{F}_i  ,\star) & \star & \star & \star  \\
        ( \bm{P}_1 \bm{B}_{2i}\bm{F}_i)^T  &    ( \bm{P}_2 \bm{A}_i-\bm{G}_{i} \bm{C}_{2i}  ,\star)   & \star & \star \\
        (\bm{P}_1 \bm{B}_1)^T &  (\bm{P}_2 \bm{B}_1)^T& -\gamma^2 \bm{I} & \star  \\
        \bm{C}_{1i}  &   0 & 0& -\bm{I}
      \end{array}
    \right]<0
}
\end{equation*}
 for  $i=1,2,3,4$.
The $\bm{F}_i,\bm{G}_i,i=1,2,3,4$, and $\bm{P}_i>0$, $i=1,2,$ are the decision variables, and
the observer gains $\bm{L}_i,i=1,2,3,4,$ are recovered by $\bm{L}_i=\bm{P}_2^{-1} \bm{G}_{i}$.
The external and internal variables were classified as
  $\bm{\alpha}=(\bm{F}_1,\bm{F}_2,\bm{F}_3,\bm{F}_4)$ and $\bm{X}=(\bm{P}_1,\bm{P}_2,\bm{G}_1,\bm{G}_2,\bm{G}_3,\bm{G}_4)$, respectively.
The bounds
\begin{equation*}
    eig\{\bm{A}_i-\bm{B}_{2i} \bm{F}_i\}\in  \{\sigma+j\omega :(\sigma,\omega ) \in [-20,0]\times [-20,20]  \}
\end{equation*}
for $i=1,2,3,4$,
were used.

\section{SOPs and MOPs with BMI Constraints}\label{app_SO}

If not specified, the following equation numbers are those in the respective references.

LPVS (Linear Parameter-varying System, (28)--(32) in Sec.~V-A of~\cite{LPVS}):
The problem was presented in~(\ref{eq_LPVS_SO}) of this paper, and the bounds
  $ \delta_i\in [0,1] $  and $ \varsigma \in [0,10]$ were used.

SSS  (Simultaneous State-feedback Stabilization, Sec.~4.2~\cite{SSS}):
A stabilizing state-feedback gain $\bm{F}$ exists if  the optimum of
\begin{equation*}
 \begin{split}
    \max_{\bm{F},\gamma_i } & \; \min \{ \gamma_1,\gamma_2,\gamma_3  \}  \\
    \mbox{subject to } &      [\bm{F}]_{mn}\leq F_{\max}\\
&               (  \bm{P}_i ( \bm{A}_i+\bm{B}_{i} \bm{F})  ,\star)+2\gamma_i  \bm{P}_i<0\\
&    \bm{P}_i>0, i=1,2,3
 \end{split}
\end{equation*}
is positive.
We let  $\bm{\alpha}=(\gamma_1,\gamma_2,\gamma_3,\bm{F})$, $\bm{X}=(\bm{P}_1,\bm{P}_2,\bm{P}_3)$, and $\bm{\mathcal{F}}(\bm{\alpha} )= - \min \{ \gamma_1,\gamma_2,\gamma_3  \} $.
The bounds
 $  [\bm{F}]_{mn} \in [-50,50]$ and  $\gamma_i \in [0,5]$
were used.

MCD  (Mixed $H_2$/$H_\infty$ Controller Design, Sec.~III-A of~\cite{ostertag2008improved} with BMIs in (3)):
We let $\bm{\alpha}=(\eta,\bm{K})$ and $\bm{X}=(\bm{P}_1,\bm{P}_2)$.
The bounds  $ \eta \in [0,2] $  and $ [\bm{K}]_{n} \in [-5,5]$ were used.

\emph{Spectral Abscissa Optimization}:
The bounds $ [\bm{F}]_{mn} \in [-50,50] $  and $   eig\{\bm{A}_{ \bm{F} }\}    \in  \{\sigma+j\omega :(\sigma,\omega ) \in [-20,0]\times [-20,20]  \}$
were used, $S=3$ subspaces were adopted, and  $\kappa_s $ was chosen from~$\{1,0.5,0.1\}$ uniformly at random
when each~$\bm{\alpha}_i$ in~(\ref{eq_ini_pop}) of this paper was constructed.
The same setting was used in the $H_2$ optimization, $H_\infty$ optimization, and MOPs as well.

\emph{$H_2$ Optimization}:
We let $\bm{D}_{11}=0$ in $G_{c\ell}(\bm{F})$ and solved
\begin{equation*}
 \begin{split}
    \min_{\bm{Y},\bm{F},\bm{Q}  } & \; ||G_{c\ell}(\bm{F})||_2  \\
    \mbox{subject to }  & \;   ( \bm{A}_{\bm{F}}\bm{Q},\star)+  \bm{B}_1 \bm{B}_1^T <0 \\
      & \;
  \left[
   \begin{array}{cc}
     \bm{Y}  & \bm{C}_1{\bm{Q}} \\
     \star &  \bm{Q} \\
   \end{array}
 \right]>0,  \bm{Q}>0.
 \end{split}
\end{equation*}
The term $\bm{B}_1 \bm{B}_1^T$ was replaced by $\bm{B}_1 \bm{B}_1^T+ 10^{-5} \bm{I}$ if it was not positive definite.
We let~$\bm{\alpha}=\bm{F}$, $\bm{X}=(\bm{Q},\bm{Y})$, and
$\bm{\mathcal{F}}(\bm{\alpha} )=||G_{c\ell}(\bm{F})||_2$ according to the classification principle.
Given the value of the external variable~$\bm{\alpha}=\bm{F}$, existing deterministic algorithms can be applied to evaluate~$\bm{\mathcal{F}}(\bm{\alpha} )=||G_{c\ell}(\bm{F})||_2$, e.g.,
 the MATLAB routine $\mathrm{norm(syst,p)}$ with $\mathrm{syst}=G_{c\ell}(\bm{F})$  and  $\mathrm{p}=2$ can be used.
 To facilitate numerical comparisons and expedite the solving process, we used
the following settings for both $H_2$ and $H_\infty$ optimization:
 the objective values $\tilde{\bm{\mathcal{F}}}:=[ 10^5\; 10^5+\alpha_o(\bm{A}_{\bm{F}})]^T$ were assigned without further evaluation of $\lambda^*(\bm{\alpha})$
whenever $\alpha_o(\bm{A}_{\bm{F}})\geq 0$; and
the bounds
$   eig\{\bm{A}_{ \bm{F} }\}    \in  \{\sigma+j\omega :(\sigma,\omega ) \in [-20,0]\times [-20,20]  \}$
were not used (only bounds on $[\bm{F}]_{mn}$ were used)  when problem NN11 in COMPl$_e$ib was solved.

 \emph{$H_\infty$ Optimization}:
We solved
\begin{equation*}
 \begin{split}
    \min_{\bm{Y},\bm{F},\gamma  } & \;  ||G_{c\ell}(\bm{F})||_\infty  \\
    \mbox{subject to }  & \;
 \left[
   \begin{array}{ccc}
    ( \bm{Y}\bm{A}_{\bm{F}},\star) & \bm{XB}_1 & \bm{C}_{\bm{F}}^T \\
     \star & -\gamma \bm{I} & \bm{D}_{11}^T \\
    \star & \star& -\gamma \bm{I}\\
   \end{array}
 \right]<0\\
 & \;  \bm{Y}>0, \gamma >0.
     \end{split}
\end{equation*}
We let~$\bm{\alpha}=\bm{F}$, $\bm{X}=(\gamma,\bm{Y})$, and
$\bm{\mathcal{F}}(\bm{\alpha} )=||G_{c\ell}(\bm{F})||_\infty$.
Given the value of the external variable~$\bm{\alpha}=\bm{F}$, the value of~$\bm{\mathcal{F}}(\bm{\alpha} )=||G_{c\ell}(\bm{F})||_\infty$
can be determined using deterministic algorithms, e.g.,  the MATLAB routine $\mathrm{norm(syst,p)}$ with $\mathrm{syst}=G_{c\ell}(\bm{F})$  and $\mathrm{p=inf}$ can be used.

 \emph{MOPs}: We used
the bound~$\beta\in [0,1.5]$
in the sparse linear constant output-feedback design,
and the bounds $ \eta \in [0,2],\gamma \in[1,5]  $
 in the mixed $H_2$/$H_\infty$ design.


\end{document}